%% file: plain.tex
\documentclass{article}[11pt]

\usepackage[letterpaper, margin=1.1in]{geometry}
\usepackage{amsmath,amssymb,amsthm}
\usepackage{tikz-cd,color,comment}
\usetikzlibrary{arrows}

\input{macros}

\begin{document}

\title{ Using Symmetry to Schedule Classical Matrix Multiplication}

\author{
Harsha Vardhan Simhadri\\
Lawrence Berkeley National Lab\\
\texttt{harshas@lbl.gov}
}
\maketitle
\input{abstract}



\vfill

\renewcommand{\vspace}[1]{}

\input{intro}
\input{model}
\input{comm-sol}
\input{matmul}

\input{conclusion}

\bibliographystyle{abbrv}
\bibliography{ref}  

\appendix
\input{appendix}

\end{document}

%% file: macros.tex
\newtheorem{theorem}{Theorem}
\newtheorem{definition}[theorem]{Definition}
\newtheorem{proposition}{Proposition}
\newtheorem{lemma}{Lemma}

\newcommand{\sym}[1]{\mathfrak{S}_{#1}}
\newcommand{\perm}[1]{S_{#1}}
\newcommand{\permtwo}[2]{\perm{#1}\times\perm{#2}}
\newcommand{\permthree}[3]{\perm{#1}\times\perm{#2}\times\perm{#3}}
\newcommand{\Shift}[1]{\Sigma_{#1}}
\newcommand{\ints}{\mathbb{Z}}
\newcommand{\torus}[2]{\left(\ints/{#1}\ints \right)^{#2}}
\newcommand{\intp}[1]{\ints/{#1}\ints}
\newcommand{\iterwr}[2]{\perm{#1}^{\wr{#2}}}

\newcommand{\homsched}{\rho}
\newcommand{\homloc}{\rho_d}
\newcommand{\hominp}{\rho_l}
\newcommand{\homtime}{\rho_T}

\newcommand{\shift}{\sigma_{\to}}
\newcommand{\add}[1]{+_{#1}}
\newcommand{\subt}[1]{-_{#1}}
\newcommand{\gen}[1]{\langle{#1}\rangle}

\newcommand{\tinc}{\Delta}

\newcommand{\loc}[1]{l_{#1}}
\newcommand{\sched}{f}
\newcommand{\inp}[1]{d_{#1}}
\newcommand{\timepr}[1]{T_{#1}}

\newcommand{\defn}[1]{{\textit{\textbf{\boldmath #1}}}}

\newenvironment{closeitemize}
    {\begin{list}{$\bullet$}
    {
         \setlength{\itemsep}{-0.3\baselineskip}
     \setlength{\topsep}{0.15\baselineskip}
     \setlength{\parskip}{0pt}}}
    {\end{list}}
\newcounter{ccount}
\newenvironment{closeenum}
    {\begin{list}{\arabic{ccount}.}
    {\usecounter{ccount}\setlength{\itemsep}{-0.3\baselineskip}
     \setlength{\topsep}{0.15\baselineskip}
     \setlength{\parskip}{0pt}}}
    {\end{list}}

%% file: abstract.tex
\begin{abstract}
Presented with a new machine with a specific interconnect topology,
algorithm designers use intuition about the symmetry of the algorithm
to design time and communication-efficient schedules that map the
algorithm to the machine.  Is there a systematic procedure for
designing schedules?  We present a new technique to design schedules
for algorithms with no non-trivial dependencies, focusing on the
classical matrix multiplication algorithm.

We model the symmetry of algorithm with the set of instructions $X$ as
the action of the group $\sym{X}$ formed by the compositions of
bijections from the set $X$ to itself.  We model the machine as the
action of the group $N\times \tinc$, where $N$ and $\tinc$ represent
the interconnect topology and time increments respectively, on the set
$P\times T$ of processors iterated over ``time steps''.  We model
schedules as ``symmetry-preserving'' {\it equivariant} maps between
the set $X$ and a subgroup of its symmetry $\sym{X}$ and the set
$P\times T$ with the symmetry $N\times\ints$.  Such equivariant maps
are the solutions of a set of algebraic equations involving group
homomorphisms.  We associate time and communication costs with the
solutions to these equations.

We solve these equations for the classical matrix multiplication
algorithm and show that equivariant maps correspond to time- and
communication-efficient schedules for many topologies. We recover well
known variants including the Cannon's algorithm \cite{Cannon} and the
communication-avoiding ``2.5D'' algorithm \cite{SD11} for toroidal
interconnects, systolic computation for planar hexagonal VLSI arrays
\cite{KungVLSIalgos}, recursive algorithms for fat-trees
\cite{Leiserson85}, the cache-oblivious algorithm for the ideal cache
model \cite{FLPR99}, and the space-bounded schedule
\cite{CSBR13,BFGS11} for the parallel memory hierarchy model
\cite{ACF93}.  This suggests that the design of a schedule for a new
class of machines can be motivated by solutions to algebraic
equations.

\end{abstract}

%% file: intro.tex
\section{Introduction}
Parallel computers are varied in their network and memory
characteristics so that no single version of an algorithm for a
problem works uniformly well across all machines.  This poses a
recurring challenge to algorithm designers who are forced to engineer
``new algorithms'' for a target machine.  Research literature as well
as textbooks on parallel algorithms~\cite{Leighton} list dozens of
variants for such algorithms that are suited to different network
topologies and memory hierarchies. Algorithms for fundamental problems
like matrix multiplication are re-engineered for every new
architecture.

Many of these ``new algorithms'' are often simply a new schedule or a
different data layout of a well known algorithm. The redesign
optimizes the location and the order of execution of the instructions
in the algorithm to minimize communication and running time on the
target machine.  This machine-specific schedule design is, more often
that not, done through human intuition about the symmetry of the
algorithm and the machine, as opposed to a systematic or automated
procedure.  We address this gap in the case of algorithms with no
non-trivial dependencies -- algorithms whose instructions can be
executed in any order (e.g. matrix multiplication, direct n-body
methods, tensor contractions).  We propose models for capturing the
symmetry of such algorithms and machine characteristics (network
topology and memory hierarchy), and show how to use these models
systematically.

\textbf{Our approach} is inspired by three simple observations.
\textbf{First,} the transformations which leave a set invariant form a
group, called the symmetry group of the set. The transformations can
be seen as the ``actions'' of this group on the set.  Therefore, we
model the symmetry of an algorithm as the group $\sym{X}$ of the
symmetries of its instruction set $X$.

\textbf{Second}, many machines can be described as the action of a
group that models the interconnect ($N$) on the set of nodes some of
which have processing elements ($P$).  Similarly time can be modeled
as the action of a group representing time increments ($\tinc$) on the
set representing time steps ($T$).  Putting these two together, the
movement of data over the network between time steps can be modelled
as the action of the group $N\times \tinc$ on the set $P\times T$.

\textbf{Third}, the schedules of common variants of algorithms seem to
be ``symmetry-preserving'' maps from the instruction set $X$ to the
set $P\times T$ that preserve some subset of the symmetries of
$X$. Such schedules can be modelled as ``equivariant'' maps from $X$
to $P\times T$ that commute with the group actions on these sets. 
\footnote{can also be seen as graph homomorphisms; see
  Sec.~\ref{sec:sched-model}.}  That is, the equivariant map and an
action of the symmetry group can be applied to an element $x\in X$ in
any order with the same result (Fig.~\ref{fig:GH-cd}).

Therefore, we pose the problem of finding a machine-specific schedule
and data layout for an algorithm as one of finding equivariant maps,
thus narrowing our search significantly. We also associate time and
communication costs with these solutions (Section~\ref{sec:model}).
The problem is then reduced to finding the optimal solution to an
instance of algebraic equations.  The solutions to these equations
correspond to homomorphisms between subgroups of the symmetry group
$\sym{X}$ and the group $N\times \tinc$
(Section~\ref{sec:comm-sol}). We use knowledge of the structure of
these groups to enumerate and optimize over feasible solutions.

We demonstrate the effectiveness of this technique with the example of
the classical matrix multiplication algorithm.  The instructions of
this algorithm are indexed by three arrays. Assuming that the addition
is commutative, the indices in each array may be permuted in any
order.  Therefore, the symmetry group of the matrix multiplication
algorithm can be expressed as a product of three permutation groups
(group consisting of permutation of an array) corresponding to the
three indices. The subgroups of a permutation group (a.k.a. symmetric
group) are well-understood \cite{ONanScott}.  We use this knowledge to
derive cost-efficient schedules for various machines, recovering
several well-known variants of classical matrix multiplication
(Section~\ref{sec:matmul}).

Although individual technical results in this paper are not radically
new, the \textbf{principal contribution} of the paper is a fresh and
unifying perspective to the problem of schedule design.  The tools and
the line of reasoning developed in this paper can be used to
systematically derive schedules for future architectures or reason
about different data layouts.

\textbf{Related Work}.  One approach to systematically adapt
algorithms to machines consists of the oblivious paradigms
\cite{BPPS14,CSBR13,BFGS11} which propose automatic---even
online---schedules for machines with tree-like interconnects. The
machines these models represent can be shown to be capable of
executing efficient variants of any parallel algorithm competitively
compared to any other architecture with the same hardware resource
constraints (such as VLSI area).  Algorithms that are provably
(asymptotically) optimal for such machines when coupled with these
schedules can be designed \cite{CR07,BGS10,harsha-thesis,BPPS14} and
achieve competitive performance. However, this approach can not
generate the best schedule for any topology.

Another tool for developing schedules is the Polyhedral Model (also
called the Polytope model) \cite{Lengauer93,Feautrier96}. This model
has its roots in the analysis of recurrence equations \cite{KMW67} and
automatic parallelization of ``DO-loops'' \cite{Lamport74}.  The
principal aim of the model was parallelization of nested loops for
vector architectures. The problem was modeled as Parallel Integer
Programming \cite{Feautrier88} giving rise to many code-generation
tools \cite{AI92,QRW00}.  Loop transformation techniques based on
affine transformations for improving the performance of nested loops
were proposed \cite{WL91-TPDS,LimLam97}, and blocking for locality was
considered \cite{WL91}.  A number of features were incrementally added
to the Polyhedron model including the notion of multi-dimensional time
\cite{Feautrier92-multitime}, data replication for reducing
communication and synchronization \cite{CLS99}, and memory management
\cite{QR00,DSV05,LF98}.  Some of these results are summarized in
\cite{Gri04}. Similar efforts for adapting algorithms to systolic
arrays include \cite{Moldovan82,KL83,Quinton87}.

Our work differs from previous work in that it is based on groups and
their homomorphisms rather than affine transformations which are the
basis for most previous work. Since groups are more expressive, our
model elegantly captures many of the features that were incrementally
retrofitted to earlier models such as data layout, data movement,
schedules, multiple copies of data and multi-dimensional time. The
model is succinct enough to be represented in a diagram
(Fig.~\ref{fig:maincd}) which we will proceed to explain.

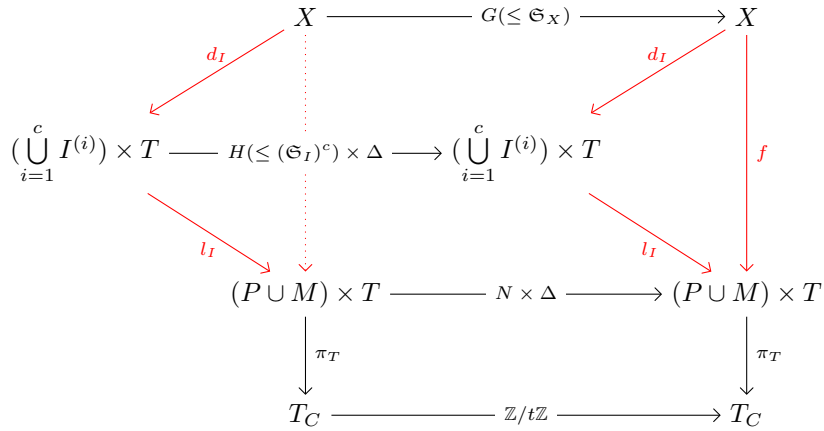
\begin{figure}[!tb]
  \centering
\begin{tikzpicture}[descr/.style={fill=white}]
  \matrix (m) [matrix of math nodes, row sep=3em, column sep=2em]
{
    & X & & X \\
    (\bigcup\limits_{i=1}^{c}I^{(i)})\times T & & (\bigcup\limits_{i=1}^{c}I^{(i)})\times T & \\
    & (P\cup M) \times T & & (P\cup M) \times T \\
    & T_C &  & T_C  \\
};
  \path[-stealth,->,>=angle 90,font=\scriptsize]
    (m-1-2) edge node[descr] {$G(\leq \sym{X})$} (m-1-4)
            edge [red] node[above] {$\inp{I}$} (m-2-1)
            edge [red,dotted] node[auto] {$\sched$} (m-3-2)
    (m-1-4) edge [red] node[right] {$\sched$} (m-3-4)
            edge [red] node[above] {$\inp{I}$} (m-2-3)
    (m-2-1) edge node[descr] {$H(\leq (\sym{I})^c)\times\tinc$} (m-2-3)
            edge [red] node[below] {$\loc{I}$} (m-3-2)
    (m-2-3) edge [red] node[below] {$\loc{I}$} (m-3-4)
    (m-3-2) edge node[descr] {$N\times\tinc$} (m-3-4)
            edge node[right] {$\pi_T$} (m-4-2)
    (m-4-2) edge node[descr] {$\intp{t}$} (m-4-4)
    (m-3-4) edge node[right] {$\pi_T$} (m-4-4);
\end{tikzpicture}
  \caption{ A commutative diagram summarizing all the
    symmetry-preserving maps described in Section~\ref{sec:model}.
    The set $X$ represents the instructions in the algorithm, the set
    $(\cup_{i=1}^{c}I^{(i)})\times T$ represents $c$ copies of the
    input set $I$ iterated over time, and $(P\cup M)\times T$
    represents the set of processors and memory units iterated over
    time.  The group action over each of these sets is depicted with a
    text-overlaid arrow.  The equivariant maps between these sets are
    shown in red and are to be solved for.  These are the functions
    $\sched$, which determines the schedule, $\loc{I}$, which
    determines the data placement, and $\inp{I}$, which determines the
    time and the copy of variable in each variable set used for an
    instruction.  Similar equations can be written for other input and
    output sets.}
  \label{fig:maincd}
\end{figure}


%% file: model.tex
\section{Models for Computation, Machine and Symmetry-Preserving Maps }
\label{sec:model}
We present models for the symmetries of an algorithm in which
instructions can be executed in any order (Sec.~\ref{sec:comp-model}),
and the topology of the target machine (Sec.~\ref{sec:machine-model}).
Schedules are defined as maps from the instruction set to the set of
processors iterated over time.  We will formalize the notion of
symmetry-preserving maps and represent it as commutative diagrams.  We
also incorporate data placements and data movements in to these models
(Sec.~\ref{sec:sched-model}). We assign costs to the maps
corresponding to schedules and data movement (Sec.~\ref{sec:costs})
and add more features to the model (Sec.~\ref{sec:extension}).

\subsection{Computation model and symmetries}
\label{sec:comp-model}
An algorithm consists of a finite set of instructions $X$ and a finite
set of input and output variable sets.  Each instruction accesses one
variable from each input and output variable set so that $X$ is a
subset of the direct product of the variable sets. For classical
matrix multiplication the input variable sets are the variables
denoting the entries of the two input matrices, and the output
variable set corresponds to output matrix.  Let $[n]$ denote the set
$\{0,1,\dots,n-1\}$.  For a size $l\times m \times n$ multiplication,
the input variable sets are $A=\{A_{rr'}\}_{r\in[l],r'\in[m]}$ and
$B=\{B_{rr'}\}_{r\in[m],r'\in[n]}$, and the output variable set is
$C=\{C_{rr'}\}_{r\in[n],r'\in[l]}$. The set of instructions $X$ is
$\{(A_{ij},B_{jk},C_{ki})\}_{i\in[l],j\in[m],k\in[n]} \subset A\times
B\times C$.

Groups are a standard way to model the symmetries of a set (see
Appendix~\ref{sec:prelims} for definitions of groups, homomorphisms
and group actions).  Consider all bijections (invertible functions)
from the set $X$ to itself. These operations form a group under
composition. This is referred to as the \defn{symmetry group}
$\sym{X}$ of $X$. The symmetry group $\sym{X}$ acts naturally on the
set $X$: the action $\sigma\cdot$ of each element $\sigma\in\sym{X}$
is the bijective map that $\sigma$ represents. We denote this group
action diagrammatically by
\begin{center}
\begin{tikzcd}[row sep=0em,column sep=9em]
  X \arrow{r}[description]{\sym{X}} & X.
\end{tikzcd}
\end{center}
The group action can also be seen as the \textbf{graph} with vertex
set $X$ and edges $x\mapsto \sigma\cdot x$ for each $x\in X$ and
$\sigma\in\sym{X}$.
  
In the case of matrix multiplication, the bijections correspond to
separate permutations on the indices $i\in[l]$, $j\in[m]$ and
$k\in[n]$ (we assume \texttt{+=} in $C_{ki}$\texttt{+=}$A_{ij}\cdot
B_{jk}$ is associative and commutative).  The group formed by the
permutations of $n$ numbers (or symbols) under the composition
operation is called the \defn{symmetric group} of $n$ numbers, and
denoted by $\perm{n}$.  Therefore, the symmetry group of the set of
instructions of classical matrix multiplication algorithm is
isomorphic to the direct product of three symmetric groups $\perm{l}$,
$\perm{m}$ and $\perm{n}$: $\sym{X} \cong
\perm{l}\times\perm{m}\times\perm{n}$.  So we have
\begin{center}
\begin{tikzcd}[row sep=0em,column sep=7em]
{\{(A_{ij},B_{jk},C_{ki})\}_{i\in[l],j\in[m],k\in[n]}} =: X
\arrow{r}[description]{\permthree{l}{m}{n}} & X.
\end{tikzcd}
\end{center}
where $\perm{l}$, $\perm{m}$ and $\perm{n}$ act on indices $i, j$ and
$k$ respectively. In this case, the graph corresponding to the group
action has vertices $X_{ijk}$ indexed by $i,j$ and $k$.  If a triplet
of permutations on $[l],[m]$ and $[n]$ take the indices $i$ to $i'$,
$j$ to $j'$ and $k$ to $k'$ respectively, then there is an edge
between vertices $X_{ijk}$ and $X_{i'j'k'}$ labelled with this triplet
of permutations.

We can similarly model the symmetry of an input set $I$:
\begin{center}
\begin{tikzcd}[row sep=0em, column sep=9em]
  I \arrow{r}[description]{\sym{I}} & I.
\end{tikzcd}
\end{center}
For matrix multiplication, we have for the input set $A$,
\begin{center}
\begin{tikzcd}[row sep=0em,column sep=9em]
{\{A_{rr'}\}_{r\in[l],r'\in[m]}} =: A
 \arrow{r}[description]{\permtwo{l}{m}} & A
\end{tikzcd}
\end{center}
where $\perm{l}$ and $\perm{m}$ act on the indices $r$ and $r'$
respectively.

\subsection{Models for machine topology}
\label{sec:machine-model}
A machine consists of a set of nodes connected by a network. Each node
consists of a finite amount of memory. Some nodes also have a
sequential processing element and we refer to such a node as a
processor.  For simplicity, we will deal with the case where all nodes
are processors until Sec.~\ref{sec:extension}.  Just as in the
computation model, we will model the machine as a group action on a
set.

We start with the simple example of a ``2D-torus'' machine consisting
of $q\times q$ processors connected with a 2D-torus network.  The
network can be modeled as the group $\torus{q}{2}$ where
$(\ints/q\ints)$ is the group of integers $\{0,1,\dots,q-1\}$ with the
group operation '$\add{q}$', addition modulo $q$. The group element
$(i,j)$ indicates traversing $i$ hops to the right and $j$ hops
upwards (with wrap-around), and $(-i,-j)$ represents $i$ hops to the
left and $j$ hops downwards.  The machine can be modeled as the action
of the \defn{network group} $N:=\torus{q}{2}$ on the processor set $P
= \{P_{xy}\}_{x\in[q],y\in[q]}$. The element $(i,j)\in N$ maps
$(i,j)\cdot P_{x,y} \mapsto P_{i\add{q}x,j\add{q}y}$ indicating the
relative position of these two processors.

We add the notion of time steps (not necessarily clock cycles) to the
model with the set $T$. A machine with synchronous time steps across
the processors is modeled as the set $P\times T$ representing the
processor set iterated over time steps. Time increments are modeled as
the action of the \defn{time increment group} $\tinc$ on the set
$T$.  For example, when a machine works for $t$ time steps all of the
same duration, $T$ is represented by the set $\{t_i\}_{i\in[t]}$, and
the increment by the action of $\tinc=\intp{t}$ over $T$.  The set of
possible data movements (communication) over the machine is modeled as
the action of the group $N\times \tinc$ on the set $P\times T$.  We
represent it diagrammatically by
\begin{center}
\begin{tikzcd}[row sep=0em,column sep=9em]
  {P\times T} \arrow{r}[description]{N\times\tinc} & {P \times T}.
\end{tikzcd}
\end{center}

For a 2D-torus with synchronous and uniform time steps, $((i,j),t')
\cdot (P_{x,y},t_0)\mapsto (P_{x\add{q}i,y\add{q}j},t_0\add{t}t')$
corresponds to moving all the variables at the processor $P_{x,y}$ at
time step $t$ to the processor $P_{x\add{q}i,y\add{q}j}$ at step $t'$
(Fig.~\ref{fig:comm-action} left).

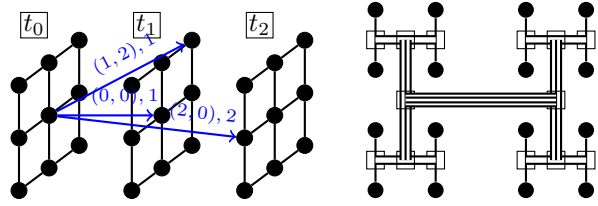
\begin{figure}
  \begin{center}
\input{fig-comm-action}
  \caption{(left) Communication model of a 2D-torus of size $3\times
    3$; group action in blue. (right) Fat-tree interconnect with 16
    processors (black circles).}
  \end{center}
\label{fig:comm-action}
\end{figure}

Often, it is natural to consider time steps that are not of the same
duration.  For instance, for tree-like machines such as in the case of
machine with fat-tree network \cite{Leiserson85}
(Fig.~\ref{fig:comm-action} right; see Sec.~\ref{sec:extension} for a
model for fat-tree) or parallel memory hierarchies \cite{ACF93}
(Fig.~\ref{fig:pmh}), it is natural to consider schedules that
progress in time steps of different granularity, one corresponding to
each level in the tree.  For this, we can model time steps as the set
$T=\timepr{1}\times \timepr{2}\times\dots \timepr{k}$, where each
$\timepr{l+1}:= \{t_{l,i}\}_{i\in[t_l]}$ represents $t_l$ steps within
a $\timepr{l}$ superstep.  The group $\tinc=\prod_{l\in[k]}\intp{t_l}$
acting on $T$ models time increments at each granularity.

\subsection{Equations for schedule and data placement}
\label{sec:sched-model}
A \defn{schedule} is a map from the set of instructions $X$ to the set
of processors iterated over time $P\times T$.  We consider those
schedules that are symmetry-preserving equivariant maps.
\begin{definition}[Equivariant maps]
Fix the action of the group $G$ on the set $X$, and the action of the
group $H$ on the set $Y$.  A map $f:X\rightarrow Y$ is
$(G,H)_\rho$-equivariant for the group homomorphism $\rho:G\rightarrow
H$ if $f(g\cdot x) = \rho(g)\cdot f(x)$ for all elements $x\in X$ and
the actions of $g\in G$ and $\rho(g)\in H$. A $(G,G)_{Id}$-equivariant
map is simply called $G$-equivariant.
\end{definition}

\begin{figure}[!htb]
  \centering
\begin{tikzcd}[row sep=3em,column sep=3em]
    X \arrow[red]{d}{f}
      \arrow{r}[description,name=Ggrp]{G}&
    X \arrow[red]{d}{f} & 
    x \arrow[red,mapsto]{d}{f}
      \arrow[mapsto]{r}[below,name=gact]{g\cdot}& 
    g\cdot x \arrow[red,mapsto]{d}{f} \\
    Y \arrow{r}[description,name=Hgrp]{H} & 
    Y & 
    y \arrow[mapsto]{r}[name=hact]{h\cdot} &
    h\cdot y 
\arrow[blue,to path={(Ggrp) -- node[auto] {$\rho$} (Hgrp)}]{}
\arrow[mapsto,blue,to path={(gact) -- node[auto] {$\rho$} (hact)}]{}
\end{tikzcd}
  \caption{Commutative diagram depicting the $(G,H)_\rho$-equivariant
    map $f$ for the homomorphism $\rho:G\rightarrow H$.  The diagram
    on the left indicates that the function $f:X\rightarrow Y$ makes
    the diagram on the right commute (different paths lead to the same
    answer) for all choices of elements $x\in X, y\in Y$ s.t. $y=f(x)$
    and the actions of elements $g\in G, h\in H$ s.t. $h=\rho(g)$.}
  \label{fig:GH-cd}
\end{figure}
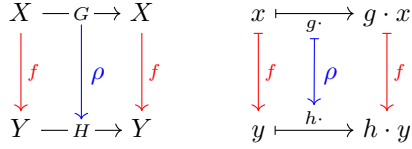

The commutative diagram above can also be seen as a graph
homomorphism.  Let $A$ and $B$ be the graphs corresponding to the
action of the group $G$ on $X$ and group $H$ on $Y$. The edge map
defined by the group homomorphism (blue arrow) and the vertex map
defined by the equivarant map (red arrow) define a graph homomorphism
from $A$ to $B$. The rest of the paper uses algebra to parameterize
such graph homomorphisms and reason about their costs.

The \textbf{schedule} $\sched$ is modeled as an $(G,N\times
\tinc)_\homsched$-equivariant map from $X$ to $P\times T$ for some
choice of subgroup $G\leq\sym{X}$ of the symmetries of $X$ and
homomorphism $\homsched:G\rightarrow N\times\tinc$
(Fig.~\ref{fig:sched-data}). The choice of subgroup $G$ and the
homomorphism to $N\times\tinc$ significantly narrows down the possible
equivariant maps to a few parameters (Sec. \ref{sec:comm-sol}).

To trace the movement of the input
variables over time, we consider the set $I \times T$, and define the
action of the group $\sym{I}\times \tinc$ on the set $I\times T$ as
the natural extension of the action of $\sym{I}$ on $I$ and the action
of $\tinc$ on $t\in T$ (similar to the case of $P\times T$).  We
denote this by
\begin{center}
\begin{tikzcd}[row sep=0em,column sep=9em]
  {I \times T} \arrow{r}[description]{\sym{I}\times \tinc} & {I \times T}.
\end{tikzcd}
\end{center}
The location of the input set $I$ on the machine is modeled as an
$(H\times \tinc,N\times\tinc)_{\homloc}$-equivariant map from the set
$I\times T$ to the set $P\times T$, for some subgroup $H\leq\sym{I}$
and homomorphism $\homloc:H\times\tinc \rightarrow N\times\tinc$ of
the form
\[
\homloc = (\mu:H\times \tinc\rightarrow N) \times Id_{\tinc},
\]
where $\mu$
reflects its movement across time steps.  Similar equations can be written for other input and output
sets.

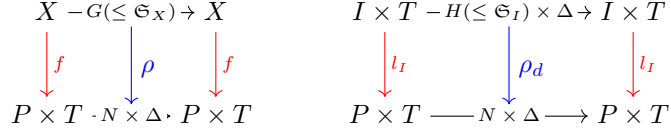
\begin{figure}[!htb]
  \centering
\begin{tikzcd}[row sep=2.5em,column sep=3em]
    X \arrow[red]{d}{\sched}
      \arrow{r}[description,name=Xact]{G(\leq\sym{X})}&
    X \arrow[red]{d}{\sched} & 
    {I\times T} \arrow[red]{d}{\loc{I}}
      \arrow{rr}[description,name=Iact]{H(\leq\sym{I})\times\tinc}& &
    {I\times T} \arrow[red]{d}{\loc{I}} \\
    {P\times T} \arrow{r}[description,name=PTact]{N\times\tinc} & 
    {P\times T} & 
    {P\times T} \arrow{rr}[description,name=Pact]{N\times\tinc} & &
    {P\times T} 
\arrow[blue,to path={(Xact) -- node[auto] {$\homsched$}(PTact)}]{}
\arrow[blue,to path={(Iact) -- node[auto] {$\homloc$}(Pact)}]{}
\end{tikzcd}
  \caption{ Commutative diagrams depicting the schedule $\sched$ and
    data placement function $\loc{I}$ in red.  }
  \label{fig:sched-data}
\end{figure}

The schedule and data placement have to be consistent with each other
so that the input or output element required by an instruction at a
particular processor and time is present there.  We model this
constraint as an $(G,H\times\tinc)_{\hominp}$-equivariant map
$\inp{I}:X \rightarrow I\times T$. For consistency, we add the
constraint $\homsched=\homloc \circ\hominp$.

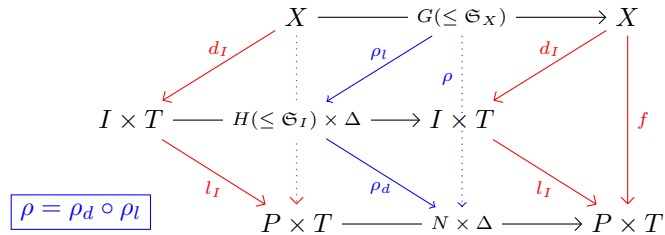
\begin{figure}[!tb]
  \centering
\begin{tikzpicture}[descr/.style={fill=white}]
  \matrix (m) [matrix of math nodes, fill=white,row sep=2.5em, column sep=3em]
{
    & X & & X \\
    I\times T & & I\times T & \\
    & P \times T & & P \times T \\
};
  \path[-stealth,->,>=angle 90,font=\scriptsize]
    (m-1-2) edge node[descr] (sched) {$G(\leq\sym{X})$} (m-1-4)
            edge [red] node[above] {$\inp{I}$} (m-2-1)
            edge [red,dotted] node[auto] {$\sched$} (m-3-2)
    (m-1-4) edge [red] node[right] {$\sched$} (m-3-4)
            edge [red] node[above] {$\inp{I}$} (m-2-3)
    (m-2-1) edge node[descr] (data) {$H(\leq\sym{I})\times\tinc$} (m-2-3)
            edge [red] node[below] {$\loc{I}$} (m-3-2)
    (m-2-3) edge [red] node[below] {$\loc{I}$} (m-3-4)
    (m-3-2) edge node[descr] (machine) {$N\times\tinc$} (m-3-4);

\draw[->,color=blue,dotted] 
(sched) -- node[left,near start,font=\scriptsize] {$\homsched$} (machine);
\draw[->,color=blue] 
(sched) -- node[above,font=\scriptsize] {$\hominp$} (data);
\draw[->,color=blue]
(data) -- node[below,font=\scriptsize] {$\homloc$} (machine);

\node[draw,blue,rectangle] at ([yshift=1.5em]current bounding box.south west){$\homsched=\homloc\circ\hominp$};
\end{tikzpicture}
  \caption{ The equivariant maps for the input set  $I$ are shown in red
    and are to be solved for.}
  \label{fig:simplecd}
\end{figure}

The commutative diagram in Fig.~\ref{fig:simplecd} represents a set of
constraints that a valid schedule and data placement must satisfy. The
principal constraint here that narrows the search for a schedule is
equivariance. Further constraints include memory limits at each
node. To find a feasible schedule and data placement, we ``solve''
this diagram for different choices of homomorphisms (see
Sec.~\ref{sec:comm-sol} for details). Each solution can be assigned
time and communication costs as in Sec.~\ref{sec:costs}.

\subsection{Time and Communication Costs}
\label{sec:costs}
A span of $t$ clock cycles, represented by the set $T_C$, is modeled
as the action of $\intp{t}$ on $T_C$.  The instructions in a schedule
can be assigned time by projecting $P\times T$ down to $T$ and further
down to clock cycles $T_C$ (the function $\pi_T$ in
Fig.~\ref{fig:costs}). The homomorphism $\homtime:N\times\tinc
\rightarrow \intp{t}$ flattens and scales time steps.  We consider
only those $\homtime$ that disregard $N$, $\homtime:n\times e_\tinc
\mapsto e_{\intp{t}}$ for all $n\in N$, so that we can overload the
notation $\homtime$ with $\homtime:\tinc\rightarrow\intp{t}$.  If we
have $\tinc = \intp{t}$ and each time step is five processor cycles
long, we have $\homtime: x\mapsto 5x$. We can flatten $k$ nested
levels of supersteps with $2$ steps per level, that is
$T=\prod_{l\in[k]} \{t_{l,0},t_{l,1}\}$ and $\tinc=(\intp{2})^k$,
using $\homtime:(\intp{2})^k \rightarrow \intp{2^k}$ that maps
$\homtime:(0,.._{i-1},1,0,..)\mapsto 2^{i-1}$ so that the $l$-th level
superstep lasts for $2^l$ clock cycles. We can choose any combination
of stretching and flattening to model the number of clock cycles
needed for communication between supersteps.

\begin{figure}[!htb]
  \centering
\begin{tikzcd}[row sep=2em,column sep=3em]
    X \arrow[red]{d}{\sched}
      \arrow{r}[description]{G(\leq\sym{X})}&
    X \arrow[red]{d}{\sched} & 
    X \arrow[red]{d}{\inp{I}}
      \arrow{rr}[description]{G(\leq\sym{X})}&&
    X \arrow[red]{d}{\inp{I}} \\
    {P\times T} \arrow{d}{\pi_T}
      \arrow{r}[description,name=time]{N\times\tinc} & 
    {P\times T} \arrow{d}{\pi_T}& 
    {I\times T} \arrow[red]{d}{\loc{I}}
      \arrow{rr}[description,name=Iact]{H(\leq\sym{I})\times\tinc}& &
    {I\times T} \arrow[red]{d}{\loc{I}} \\
    T_C \arrow{r}[description,name=flattime]{\intp{t}} &
    T_C &
    {P\times T} \arrow{rr}[description,name=Pact]{N\times\tinc} & &
    {P\times T}
\arrow[blue,to path={(time) -- node[auto,font=\scriptsize] {$\homtime$}(flattime)}]{}
\arrow[blue,to path={(Iact) -- node[auto,font=\scriptsize] {$\homloc$}(Pact)}]{}
\end{tikzcd}
  \caption{Time can be traced with the function $\sched\circ\pi_T$ and
    communication cost with the homomorphism $\nu$.  }
  \label{fig:costs}
\end{figure}

The function $\loc{I}$ describes the location of the variables in the
input set $I$ at different time steps. The communication cost
associated with a schedule is the cost of moving $I$ (and other input
and output variable sets) to the processor specified by $\loc{I}$
between time steps. When multiple paths are available, the cost is
calculated for a specified routing policy.

Since $\loc{I}$ is equivariant with $\homloc$ which is of the form
$\mu\times Id_{\tinc}$, the homomorphism $\mu:H\times \tinc\rightarrow
N$ can be used to trace the movement of $I$ between time steps.  The
homomorphism $\mu$ restricted to $e_{\tinc}$ reflects the layout of
the variable set at some time step.  Further, $\mu$ can be seen as a
function parameterized by time increments $\tinc$ that defines the
network group element that moves variables in the set $I$ through the
machine.  Therefore assigning communication costs to elements of $N$
(according to some routing policy) defines the cost of a schedule. We
simply add up the costs of network elements used across time steps
(this equivalence is established by the equivariance property).  When
$\mu$ depends only on $\tinc$, each element is moved between time
steps by the same $n\in N$ making it easy to calculate the
communication cost.

\subsection{Further details of the model}
\label{sec:extension}
A schedule $\sched$ requires a certain number of variables to be
present on a node at each time step. Those schedules that exceed the
\textbf{memory budget} are not considered. On the other hand, when the
memory available across the machines is larger than required for one
copy of the variable sets, replicating variables across nodes is often
necessary to design schedules that minimize communication
\cite{ITT04}. Therefore we extend the model to allow a constant number
of copies of each variable in an input set.  The symmetry of $c$
copies of input set iterated across time steps is
\begin{center}
\begin{tikzcd}[row sep=0em,column sep=9em]
{(\bigcup\limits_{i=1}^{c}I^{(i)})\times T}
 \arrow{r}[description]{(\sym{I})^c\times\tinc} & 
{(\bigcup\limits_{i=1}^{c}I^{(i)})\times T}.
\end{tikzcd}
\end{center}

Suppose a machine has nodes without a processor (call these $M$), we
model the machine as the action of $N\times \tinc$ on $(P\cup M)
\times T$. As before, the schedule $\sched$ is a
$(G,N\times\tinc)_{\homsched}$-equivariant map for some $G\leq
\sym{X}$.  The important difference is that the image of $\homsched$
must be a subset of $P\times T$.

All of the symmetries, group actions and equivariant maps are
summarized in the commutative diagram in Fig.~\ref{fig:maincd}.  We
omit the homomorphisms and depict the action of the entire group in
the text-overlaid arrow instead of just the subgroups relevant to the
chosen homomorphism.

We conclude this section with a \textbf{model for the fat-tree network}
 (Fig.~\ref{fig:comm-action} right).  A fat-tree of
size $n=2^k$ is a $k$-level balanced binary tree with $n$ processors
at the leaves.  To send a message from processor $P_i$ to $P_j$, the
fat-tree routes it up from $P_i$ to the least common ancestor of $P_i$
and $P_j$ in the tree and back down to $P_j$.  We model this
multi-level data movement as the action of a group that is the
\textit{wreath product} of the groups corresponding to per-level
movement.

\begin{definition}[wreath product]
The wreath product of two finite groups $A$ and $B$ for a specified
action of the group $B$ on the finite set $\Omega$ is denoted by
$A\wr_\Omega B$ and defined as the group action of $B$ on
$|\Omega|$-tuples of the elements of $A$: $A\wr_\Omega B:=K \rtimes
B$, where $K=\prod_{\omega\in\Omega} A_\omega$ and $\rtimes$
represents the action. When the index set $\Omega$ and the action of
$B$ on $\Omega$ are clear from the context, such as in the case of
$\Omega=[b]$ and $B=\perm{b}$, we simply write $A\wr B$.
\end{definition}

The elements of the group $\perm{a}\wr\perm{b} = (\perm{a})^b\rtimes
\perm{b}\leq\perm{ab}$ correspond to those permutations on the set
$[ab]$ that result from some permutation within each of the partitions
of the $ab$ elements into $\{0,..,a-1\},\{a,..,2a-1\},..,
\{(b-1)a,..,ab-1\}$, followed by some permutation over the partitions.

We model a fat-tree network of size $n=2^k$ with the group action of
the $k$-fold iterated wreath product
\[
\iterwr{2}{k}:= (((\perm{2}\wr\perm{2})\wr\perm{2})\wr\dots\perm{2}.
\]
The action of the elements of $\iterwr{2}{k}$ can be understood by
organizing the $n=2^k$ elements into a $k$-level balanced binary tree.
At each internal node in the binary tree, we can choose to either swap
the left and right subtree or leave them as is. Each set of choices
corresponds to a unique element in $\iterwr{2}{k}$. When considered as
the network group, the element $\sigma\in\iterwr{2}{k}$ that sends
index $i$ to the index $j$ ($i,j\in[n]$) sends the processor $P_i$ to
$P_j$: $\sigma\cdot P_i\mapsto P_j$.  Note that multiple permutations
send $P_i$ to $P_j$ for any choice of $i,j\in[n]$. They number
$2^{n-1}/n$ out of the total $2^{n-1}$ elements in this group.


%% file: fig-comm-action.tex
\begin{tikzpicture}
\tikzstyle{vertex}=[circle,draw,fill=black,minimum size=6pt,inner sep=1pt,font=\tiny]
\tikzstyle{leftsmall}=[left,font=\tiny]
\tikzstyle{rightsmall}=[right,font=\tiny]
\tikzstyle{edge} = [draw,thick,-,black]
\tikzstyle{invisible} = [draw=none,minimum size=0pt]
\tikzstyle{doubleedge} = [draw,thick,double distance=2pt,black]

\node[vertex] (x000) at (0,1.4) {};
\node[vertex] (x010) at (0,0.7) {};
\node[vertex] (x020) at (0,0) {};
\node[vertex] (x001) at (0.4,1.7) {};
\node[vertex] (x011) at (0.4,1) {};
\node[vertex] (x021) at (0.4,0.3) {};
\node[vertex] (x002) at (0.8,2) {};
\node[vertex] (x012) at (0.8,1.3) {};
\node[vertex] (x022) at (0.8,0.6) {};
\draw[edge] (x000) -- (x001);
\draw[edge] (x001) -- (x002);
\draw[edge] (x010) -- (x011);
\draw[edge] (x011) -- (x012);
\draw[edge] (x020) -- (x021);
\draw[edge] (x021) -- (x022);
\draw[edge] (x000) -- (x010);
\draw[edge] (x010) -- (x020);
\draw[edge] (x001) -- (x011);
\draw[edge] (x011) -- (x021);
\draw[edge] (x002) -- (x012);
\draw[edge] (x012) -- (x022);

\node[vertex] (x100) at (1.5,1.4) {};
\node[vertex] (x110) at (1.5,0.7) {};
\node[vertex] (x120) at (1.5,0) {};
\node[vertex] (x101) at (1.9,1.7) {};
\node[vertex] (x111) at (1.9,1) {};
\node[vertex] (x121) at (1.9,0.3) {};
\node[vertex] (x102) at (2.3,2) {};
\node[vertex] (x112) at (2.3,1.3) {};
\node[vertex] (x122) at (2.3,0.6) {};
\draw[edge] (x100) -- (x101);
\draw[edge] (x101) -- (x102);
\draw[edge] (x110) -- (x111);
\draw[edge] (x111) -- (x112);
\draw[edge] (x120) -- (x121);
\draw[edge] (x121) -- (x122);
\draw[edge] (x100) -- (x110);
\draw[edge] (x110) -- (x120);
\draw[edge] (x101) -- (x111);
\draw[edge] (x111) -- (x121);
\draw[edge] (x102) -- (x112);
\draw[edge] (x112) -- (x122);

\node[vertex] (x200) at (3,1.4) {};
\node[vertex] (x210) at (3,0.7) {};
\node[vertex] (x220) at (3,0) {};
\node[vertex] (x201) at (3.4,1.7) {};
\node[vertex] (x211) at (3.4,1) {};
\node[vertex] (x221) at (3.4,0.3) {};
\node[vertex] (x202) at (3.8,2) {};
\node[vertex] (x212) at (3.8,1.3) {};
\node[vertex] (x222) at (3.8,0.6) {};
\draw[edge] (x200) -- (x201);
\draw[edge] (x201) -- (x202);
\draw[edge] (x210) -- (x211);
\draw[edge] (x211) -- (x212);
\draw[edge] (x220) -- (x221);
\draw[edge] (x221) -- (x222);
\draw[edge] (x200) -- (x210);
\draw[edge] (x210) -- (x220);
\draw[edge] (x201) -- (x211);
\draw[edge] (x211) -- (x221);
\draw[edge] (x202) -- (x212);
\draw[edge] (x212) -- (x222);

\draw[edge,->,color=blue] (x011) -- (x111) node[pos=.7,sloped,above,font=\scriptsize] {$(0,0),1$};
\draw[edge,->,color=blue] (x011) -- (x210) node[pos=.8,sloped,above,font=\scriptsize] {$(2,0),2$};
\draw[edge,->,color=blue] (x011) -- (x102) node[pos=.6,sloped,above,font=\scriptsize] {$(1,2),1$};

\node[draw,rectangle,inner sep=1pt] () at (0.2,2.2) {$t_0$};
\node[draw,rectangle,inner sep=1pt] () at (1.7,2.2) {$t_1$};
\node[draw,rectangle,inner sep=1pt] () at (3.2,2.2) {$t_2$};
\end{tikzpicture}
\hspace{10pt}
\begin{tikzpicture}
\tikzstyle{vertex}=[circle,fill=black,minimum size=6pt,inner sep=0pt,]
\tikzstyle{leftsmall}=[left,font=\scriptsize]
\tikzstyle{edge} = [draw,thick,-,black]
\tikzstyle{invisible} = [draw=none,minimum size=0pt]
\tikzstyle{square} = [draw,rectangle,minimum size=0pt]
\tikzstyle{doubleedge} = [draw,thick,double distance=2pt,black]
\tikzstyle{doublewideedge} = [draw,thick,double distance=3pt,black]
\tikzstyle{doubledwideedge} = [draw,thick,double distance=4pt,black]
\tikzstyle{doubleswideedge} = [draw,thick,double distance=1pt,black]

\node[vertex] (c0000) at (0,0.8) {};
\node[leftsmall] at (c0000) {};
\node[vertex] (c0001) at (0.8,0.8) {};
\node[leftsmall] at (c0001) {};
\node[vertex] (c0010) at (0,0) {};
\node[leftsmall] at (c0010) {};
\node[vertex] (c0011) at (0.8,0) {};
\node[leftsmall] at (c0011) {};
\draw[edge] (c0000) -- (c0010);
\draw[edge] (c0001) -- (c0011);
\node[square] (c000) at (0,0.4) {};
\node[square] (c001) at (0.8,0.4) {};
\draw[doubleedge] (c000.center) -- (c001.center);

\node[vertex] (c0100) at (2,0.8) {};
\node[leftsmall] at (c0100) {};
\node[vertex] (c0101) at (2.8,0.8) {};
\node[leftsmall] at (c0101) {};
\node[vertex] (c0110) at (2,0) {};
\node[leftsmall] at (c0110) {};
\node[vertex] (c0111) at (2.8,0) {};
\node[leftsmall] at (c0111) {};
\draw[edge] (c0100) -- (c0110);
\draw[edge] (c0101) -- (c0111);
\node[square] (c010) at (2,0.4) {};
\node[square] (c011) at (2.8,0.4) {};
\draw[doubleedge] (c010.center) -- (c011.center);

\node[vertex] (c1000) at (0,2.4) {};
\node[leftsmall] at (c1000) {};
\node[vertex] (c1001) at (0.8,2.4) {};
\node[leftsmall] at (c1001) {};
\node[vertex] (c1010) at (0,1.6) {};
\node[leftsmall] at (c1010) {};
\node[vertex] (c1011) at (0.8,1.6) {};
\node[leftsmall] at (c1011) {};
\draw[edge] (c1000) -- (c1010);
\draw[edge] (c1001) -- (c1011);
\node[square] (c100) at (0,2) {};
\node[square] (c101) at (0.8,2) {};
\draw[doubleedge] (c100.center) -- (c101.center);

\node[vertex] (c1100) at (2,2.4) {};
\node[leftsmall] at (c1100) {};
\node[vertex] (c1101) at (2.8,2.4) {};
\node[leftsmall] at (c1101) {};
\node[vertex] (c1110) at (2,1.6) {};
\node[leftsmall] at (c1110) {};
\node[vertex] (c1111) at (2.8,1.6) {};
\node[leftsmall] at (c1111) {};
\draw[edge] (c1100) -- (c1110);
\draw[edge] (c1101) -- (c1111);
\node[square] (c110) at (2,2) {};
\node[square] (c111) at (2.8,2) {};
\draw[doubleedge] (c110.center) -- (c111.center);

\node[square] (c10) at (2.4,0.4) {};
\node[square] (c11) at (2.4,2) {};
\draw[doublewideedge] (c10.center) -- (c11.center);
\draw[edge] (c10.center) -- (c11.center);

\node[square] (c00) at (0.4,0.4) {};
\node[square] (c01) at (0.4,2) {};
\draw[doublewideedge] (c00.center) -- (c01.center);
\draw[edge] (c00.center) -- (c01.center);

\node[square] (c0) at (0.4,1.2) {};
\node[square] (c1) at (2.4,1.2) {};
\draw[doubledwideedge] (c0.center) -- (c1.center);
\draw[doubleswideedge] (c0.center) -- (c1.center);
\end{tikzpicture}

%% file: comm-sol.tex
\section{Solving Commutative Diagrams for Equivariant Maps}
\label{sec:comm-sol}

The diagram in Fig.~\ref{fig:maincd} is a set of constraints that the
equivariant maps (in red) must satisfy for some choice of
homomorphisms between subgroups of the groups whose actions are
depicted.  To find good schedules, we
\begin{closeitemize}
\item identify subgroups of the symmetry group $\sym{X}$ of $X$,
\item enumerate homomorphisms from this subgroup to other groups
  ($N,\tinc,\sym{I}$, etc.),
\item ``solve the commutative diagram'' for equivariant maps with
  respect to the homomorphism, and
\item find the map with the least time and communication cost,
  eliminating those that violate memory constraints.
\end{closeitemize}

By ``solving a commutative diagram'' such as Fig.~\ref{fig:GH-cd} for
$(G,H)_\rho$-invariant maps for a homomorphism $\rho:G\rightarrow H$
and a fixed action of group $G$ on $X$ and $H$ on $Y$, we mean
enumerating the maps $f:X\rightarrow Y$ that satisfy $f(g\cdot x) =
\rho(g)\cdot f(x)$ for all $x\in X, g\in G$. Solving a multi-level
commutative diagram such as in Fig.~\ref{fig:maincd} or
\ref{fig:simplecd} entails finding functions that make each individual
square in the diagram commute.

We will start with the solutions to the special case $H:=G$ and
$\rho:=Id$ as in Fig.~\ref{fig:G-cd}, which we simply refer to as
$G$-equivariant maps.
The general case will be similar.

\begin{figure}[!htb]
  \centering
\begin{tikzcd}[row sep=3em,column sep=3em]
    X \arrow[red]{d}{f}
      \arrow{r}[description,name = Gup]{G}&
    X \arrow[red]{d}{f} & 
    x \arrow[red,mapsto]{d}{f}
      \arrow[mapsto]{r}[below,name=gup]{g\cdot}& 
    g\cdot x \arrow[red,mapsto]{d}{f} \\
    Y \arrow{r}[description,name = Gdown]{G} & 
    Y & 
    y \arrow[mapsto]{r}[name=gdown]{g\cdot} &
    g\cdot y 
\arrow[blue,to path={(Gup) -- node[auto] {$Id$} (Gdown)}]{}
\arrow[blue,mapsto,to path={(gup) -- node[auto] {$Id$} (gdown)}]{}
\end{tikzcd}
  \caption{Commutative diagram depicting the $G$-equivariant map
    $f:X\rightarrow Y$. Given an action $g\cdot$ of elements $g\in G$,
    the function $f$ makes the diagram commute for all $g\in G$ and
    $x\in X, y\in Y$ s.t. $y=f(x)$.}  \label{fig:G-cd}
\end{figure}

For a subgroup $K\leq G$, the quotient $G/K$ is a set of \defn{cosets}
of the form $aK$ for some $a\in G$, where $aK$ is the set
$\{ak\ |\ k\in K\}$. All cosets of $K$ are of the same cardinality,
and there are $|G|/|K|$ distinct cosets.  We will see that
$G$-equivariant maps from the set $X$ to $Y$ are directly related to
``$G$-equivariant coset maps'' between cosets of subgroups of
$G$. Therefore, the crux of the problem involves maps between groups;
the sets $X$ and $Y$ themselves are somewhat secondary.  First, a few
facts about group action \cite{lang}.

\begin{proposition}
\label{prop:action}
An action of the group $G$ on the set $X$:
\begin{closeenum}
\item partitions $X$ into disjoint ``orbits'' $X_i$ that are
  closed and connected under the group action; $X=\sqcup_{i\in
  [q]}X_i$.
\item associates with each partition $X_i$ a ``stabilizer'' subgroup
  $K_i\leq G$ that ``stabilizes'' a point $x_i\in X_i$, i.e., $k\cdot
  x_i=x_i$ for all $k\in K_i$. Any point $x'\in X_i$ is equal to
  $g\cdot x_i$ for some $g\in G$ and is stabilized by the subgroup
  $gK_i g^{-1}$ which is isomorphic to $K_i$. If $G$ is abelian ('+'
  is commutative), $K_i$ stabilizes all points in $X_i$.
\item establishes an isomorphism $Gx\cong G/G_x$ between the orbit
  $Gx:=\{ g\cdot x\ |\ g\in G\}$ of $x\in X$ and the cosets of
  stabilizer $G_x$ of $x$. This allows us to associate with each $x\in
  X_i$ a unique coset $C_x$ from the cosets $G/K_i$ of the stabilizer
  group $K_i$ of some $x_i\in X_i$ such that $C_{x_i} = K_i$ and
  $C_{g\cdot x} = gC_x$ for all $g\in G$.
\end{closeenum}
\end{proposition}

Therefore, any $G$-equivariant map $f:X\rightarrow Y$ can be split in
to $G$-equivariant components $\{f_i:X_i \rightarrow
Y_{\kappa(i)}\}_{i\in[q]}$ between connected orbits $\{X_i\}_{i\in
[q]}$ of $X$ and the connected orbits $\{Y_j\}_{j\in [r]}$ of $Y$ for
some function $\kappa:[q]\rightarrow[r]$.
Since a group action establishes a bijection between elements of the
orbit and the cosets of a stabilizer of a point in the orbit, each
component map $f_i$ induces a ``$G$-equivariant coset map'' $f'_i$
between the cosets of the stabilizer $L_i$ of $x_i\in X_i$ and the
cosets of the stabilizer $K_{\kappa(i)}$ of some $y_{\kappa(i)}\in
Y_{\kappa(i)}$; $f'_i: G/L_i \rightarrow G/K_{\kappa(i)}$.
The following lemma \cite{eq-maps} enumerates all $G$-equivariant maps
between cosets (proof in Section~\ref{sec:proofs}).


\begin{lemma}
Let $L,K$ be subgroups of $G$. Let $\alpha:G/L \rightarrow G/K$ be a
  map between cosets of $L$ and $K$ such that (i) $\alpha(glL)
  = \alpha(gL)$ for all $g\in G,l\in L$ and (ii)
  $\alpha(gL)=g\alpha(L)$ for all $g\in G$, i.e. $\alpha$ is
  $G$-equivariant. The function $\alpha$ exists $\iff$ there exists
  $a\in G$ such that $\alpha(L)=aK$ and $L^a:=a^{-1}La \subseteq
  K$. Further, if it exists, $\alpha$ is uniquely determined by $a$;
  for all cosets $gL$ of $L$, $g\in G$, $\alpha:=\hat{a}:gL\mapsto
  gaK$.
\label{lem:coset-maps}
\end{lemma}

We are now equipped to \textbf{enumerate $G$-equivariant maps}.  Since
the group action fixes the bijection between elements of an orbit and
the cosets of the stabilizer group of a point in the orbit, the choice
of the $G$-equivariant coset map between a pair of orbits defines a
corresponding $G$-equivariant map between the pair of orbits (see
Fig.~\ref{fig:eq-choice}).

\begin{proposition}
\label{prop:choice}
Fix an action of the group of $G$ on the set $X$ that splits $X$ into
connected orbits $\sqcup_{i\in [q]}X_i$ such that for all $i\in[q]$,
the subgroup $L_i\leq G$ stabilizes a point in the orbit $X_i$. Fix an
action of $G$ on the set $Y$ that splits $Y$ into connected orbits
$\sqcup_{j\in [r]}Y_j$ such that for all $j\in[r]$, the subgroup
$K_j\leq G$ stabilizes a point in the orbit $Y_j$.  Any function
$f:X\rightarrow Y$ that is equivariant under this pair of group
actions is uniquely determined by the following choices.
\begin{closeenum}
\item From each orbit $X_i$, \textbf{choose} a point $x_i$
 stabilized by $L_i$ and associate it with the coset
$L_i$. Similarly \textbf{choose} a point $y_j$ in each $Y_j$ and
associate it with coset $K_j$.
\item \textbf{Choose} a map $\kappa:[q]\rightarrow[r]$ such that
  $G$-equivariant maps exist between orbit $X_i$ and $Y_{\kappa(i)}$
  for all $i\in[q]$.
\item \textbf{Choose} for each orbit $X_i$ the image of the $G$-equivariant coset map
$f'_i:G/L_i \rightarrow G/K_{\kappa(i)}$ for the coset $L_i$. This
  fixes the image of $f'_i$ for all other cosets of $L_i$ by
  lemma~\ref{lem:coset-maps}.
\end{closeenum}
 To evaluate $f$ at an arbitrary $x\in X_i$, (i) find $g\in G$
s.t. $x_i=g\cdot x$ (it exists since the orbit is connected), (ii)
compute the coset $gf'_i(L_i)$, (iii) find the element $y$ in the
orbit $Y_{\kappa(i)}$ associated with the coset $gf'_i(L_i)$. The
element $y$ is $f(x)$.
\end{proposition}

\begin{figure}[!htb]
\begin{center}
  \begin{tikzpicture}
\coordinate (c0) at (0,0);
\coordinate (c1) at (4,0);

\node at ([xshift=-4.5em]c0) {orbit $X_i$};
\node at ([xshift=4.5em]c1) {orbit $Y_i$};

    \draw(c0) circle (2.5em);
    \draw(c1) circle (2.5em);

    \path (c0) ++(-70:2.5em) node (x) {};

    \path (c0) 
          ++(50:2.5em) node (p0) {}
          ++(20:1em) node (s0) [blue]{$L_i$};

    \path (p0) ++(200:1em) node
          (x0) {$x_i$};
    \path (x) ++(150:0.8em) node
          (xlab) {$x$};

    \path (c1)
          ++(-130:2.5em) node (p1) {}
          ++(-150:1.5em) node (s1) [blue] {$f'_i(L_i)$};

    \path (c1) ++(100:2.5em) node (y) {}
          ++(-180:2.5em) node (s2) [blue] {$gf'_i(L_i)$};

    \path (y) ++(-40:0.7em) node
          (ylab) {$y$};

    \draw[->,blue] (p0) -- node[above] {$f_i'$}(p1);
    \draw[->] (x) -- node[left] {$g\cdot$}(p0);
    \draw[->,blue] (p1) -- node[left] {$g$}(y);
    \draw[->,red] (x) -- node[near start,above] {$f_i$}(y);

  \end{tikzpicture}
  \caption{Sets and their element in black. Equivariant map in
    red. Cosets and coset maps in blue.}
  \label{fig:eq-choice}
\end{center}
\end{figure}
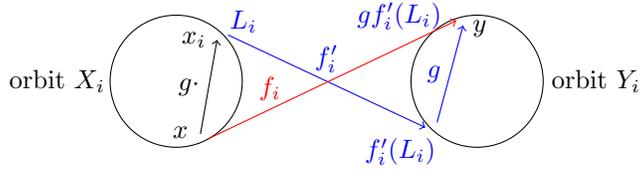

It is worth considering the special case where $G$ is abelian and
acts \defn{transitively} on $X$ (and $Y$), i.e., the action of $G$
retains $X$ (and $Y$) as one orbit.  Since $G$ is abelian, the same
subgroup stabilizes all points in the orbit.  Let $L$ and $K$ be the
stabilizers of $X$ and $Y$. The condition in
lemma~\ref{lem:coset-maps} for the existence of $G$-equivariant maps
simplifies to $L\leq K$.  The image of a point $x_0 \in X$ in $Y$
fixes the $G$-equivariant map $f:X\rightarrow Y$.  Since the points in
$Y$ are in bijection with $G/K$, such maps can be parameterized by the
cosets $G/K$.

We now consider the more general case of $(G,H)_\rho$-equivariant maps
for the homomorphism $\rho:G\rightarrow H$. The equivalent of
lemma~\ref{lem:coset-maps} for $(G,H)_\rho$-equivariant coset maps is
lemma~\ref{lem:GH-coset-map} whose proof is entirely similar to
lemma~\ref{lem:coset-maps}.

\begin{lemma}
  Let $L$ and $K$ be subgroups of $G$ and $H$ respectively. Let
  $\rho:G\rightarrow H$ be a homomorphism. Let $\alpha:G/L \rightarrow
  H/K$ be a map between cosets of $L$ and $K$ such that (i)
  $\alpha(glL) = \alpha(gL)$ for all $g\in G,l\in L$ and (ii)
  $\alpha(gL)=\rho(g)\alpha(L)$ for all $g\in G$, i.e. $\alpha$ is
  $(G,H)_\rho$-equivariant. The function $\alpha$ exists $\iff$ there
  exists $a\in H$ such that $\alpha(L)=aK$ and
  $L^a_{\rho}:=a^{-1}\rho(L)a \subseteq K$. Further, if it exists,
  $\alpha$ is uniquely determined by $a$; for all cosets $gL$ of $L$,
  $g\in G$, $\alpha:=\hat{a}:gL\mapsto \rho(g)aK$.
\label{lem:GH-coset-map}
\end{lemma}

As before, this lemma directly results in the characterization of the
maps $f:X\rightarrow Y$ that are equivariant with respect to the
transitive action of $G$ on $X$ and $H$ on $Y$.  Fix the transitive
action of $G$ on the set $X$ and suppose $L$ stabilizes some point in
$X$. Similarly, fix the action of $H$ on $Y$ and suppose $K$
stabilizes some point in $Y$.  Associate with each point in $X$ a
coset of $L$. Similar for $Y$.  Let $x_0\in X$ be a point stabilized
by the subgroup $L\leq G$.  Choosing the image of $L$ (say $aK$)
determines the image of $x_0$ ($f(x_0)$ is the point in $Y$
associated with coset $aK$) and fixes the image of every other
point. Therefore, again the $(G,H)_\rho$-equivariant maps,
if they exist, are in bijection with the cosets of $H/K$.

These observations can be extended to the case of non-transitive
actions with a treatment similar to that in the case of
$G$-equivariant maps. We simply construct functions by putting
together equivariant maps between orbits.

Finally, we make the following observation about the size of the
preimage of an equivariant map, which we use to find the memory
requirements of data placement function for a variable set $I$,
$\loc{I}$, which is equivariant with $\homloc$.
\begin{proposition}
Let $f:X\rightarrow Y$ be an $(G,H)_\rho$-equivariant map for
transitive actions of finite groups $G$ and $H$ with stabiliizers $L$
and $K$ respectively. The number of elements of $x\in X$ that map to
any given $y\in Y$ is $|G/L|/|H/K|$.
\label{prop:preimage}
\end{proposition}



%% file: matmul.tex
\section{Example: Matrix Multiplication}
\label{sec:matmul}
Using classical matrix multiplication as an example, we will
illustrate the solution of equations in Fig.~\ref{fig:maincd} for
different topologies, recovering well-known variants of matrix
multiplication.  We pick this as our example because its symmetry
group is isomorphic to the direct product of permutation groups:
$\sym{X}\cong\perm{l}\times\perm{m} \times\perm{n}$ for $l\times
m\times n$-size multiplication.
\footnote{The description of matrix multiplication as a composition of
  two linear maps between vector spaces has a substantially larger
  continuous symmetry group.  Faster (sub-cubic) algorithms use this
  fact. We do not consider these in this paper.}  The maximal
subgroups of the permutation group are known due to the O'Nan-Scott
theorem \cite{ONanScott}.

Among these, the subgroups $\perm{a}\wr\perm{b}$ for $ab=n$ are of
interest for the choice of network topologies considered in this
paper.  Further, they are transitive; for every $i,j\in[n]$, there is
an element (permutation) in the subgroup whose action takes $i$ to
$j$. Therefore when $n=2^k$, the subgroup $\iterwr{2}{k}\leq\perm{n}$
is also a transitive subgroup.

Before considering specific machine topologies, we will state a few
simple facts about homomorphisms $\homsched: G\rightarrow \intp{q}$
from some subgroup $G\leq\perm{q}$ to $\intp{q}$, when $q$ is a prime.
A permutation $\sigma\in\perm{q}$ can be seen as a collection of
disjoint permutations over partitions of $[q]$ (referred to as the
cycle decomposition). We refer to those permutations that can be
decomposed into permutations over non-trivial partitions of $[q]$ as
imprimitive, and primitive otherwise.

\begin{lemma}
Let $q$ be a prime. If $\homsched: G\rightarrow \intp{q}$ is a
homomorphism for some subgroup $G\leq\perm{q}$, all imprimitive
permutations in $G$ are in $\ker \rho:= \{g\in G\ |\ \rho(g)=e_{\intp{q}}\}$.
\label{lem:imp-hom}
\end{lemma}

See Section~\ref{sec:proofs} for proof of lemmas stated here.

Not all subgroups of $\perm{q}$ allow a non-trivial homomorphism to
$\intp{q}$. Specifically, two different kinds of primitive
permutations can not map to non-identity elements. The domains of
homomorphisms to $\intp{q}$ can contain only one kind of primitive
permutation -- permutations resulting from repeated composition of a
primitive permutation.

\begin{lemma}
Let $\rho: G\rightarrow \intp{q}$ be a non-trivial homomorphism for
some subgroup $G\leq\perm{q}$.  Let $\sigma$ be a primitive
permutation in $G$.  If $\sigma\notin\ker\rho$ and $q$ is a prime,
then all the elements in $G$ are of the form $\sigma^k$ for some
$k\in\ints$.
\label{lem:pp-hom}
\end{lemma}

Therefore, when $q$ is a prime, all the non-trivial homomorphisms
$\rho:G\rightarrow\intp{q}$ from subgroups $G\leq\perm{q}$ are
parameterized by a primitive permutation $\sigma$ and its image in
$\intp{q}$. The subgroup $G$ is generated by $\sigma$ and
$\rho:\sigma^k\mapsto \rho(\sigma)^k$.  Among primitive permutations,
cyclic permutations, i.e., permutations that cyclically shift elements
are very useful. We will denote the one-step shift by $\shift:i\mapsto
i\add{q}1$.

\begin{lemma}
  Let $\rho:G\rightarrow\intp{t}$ be a non-trivial homomorphism for
  some $G\leq\perm{q}$. If $q$ is a prime, then $q$ divides $t$.
\label{lem:tp-hom}
\end{lemma}


\subsection{2D-Torus: variants of Cannon's algorithm}
Let us start with the case $l=m=n=q$ on a $q\times q$ sized
2D-torus. Suppose that each node has three units of memory, one for
each variable from the input and output sets $A,B,C$. This limits each
node to holding one copy of each variable across the machine, and each
processor to executing one instruction per time step. So we consider
steps of one clock cycle duration with $\tinc=\intp{t}$ where $t$ is
some multiple of $q$ (lemma~\ref{lem:tp-hom}).  We will derive a
family of schedules that are cost-efficient for all values of the
parameter $q$, including prime valued $q$. We will assume $q$ is a
prime in the rest of this subsection for ease of analysis.  However,
the schedules we derive will also work for non-prime $q$.

The only subgroups of $\perm{q}$ with non-trivial homomorphisms to
$\intp{q}$ when $q$ is prime are those generated by a primitive
permutation (lemma~\ref{lem:pp-hom}). 
\footnote{It so happens here that transitive group actions suffice to
  give cost-efficient algorithms. We could also consider schedules
  arising from non-transitive group actions in general.}  Among these,
we use the transitive subgroup $\Shift{q}$ generated by the one-step
cyclic shift $\shift$ so that we choose $\Shift{q}^3\leq\sym{X}$. Note
that $\Shift{q}\cong\intp{q}$.
\footnote{subgroups generated by other primitive permutations also
  yield schedules, but are costlier in terms of communication.}  Since
the map $\loc{A}$ is equivariant with $\homloc$, the image of
$\homloc$ should be at least $q^2t$ for $\loc{A}$ to be an embedding
(Prop.~\ref{prop:preimage}).  Since the domain of $\homsched$ is
$\Shift{q}^3$ and $\homsched=\homloc\circ\hominp$, we must choose
$\Shift{q}^2\leq \sym{A}$ for the subset of the symmetries of $A$, so
that $\homloc:\Shift{q}^2\times \intp{t}\rightarrow \torus{2}{q}
\times \intp{t}$. Thus, Fig.~\ref{fig:maincd} specializes to
Fig.~\ref{fig:mm-2d}.

\begin{figure}[t]
  \centering
\begin{tikzpicture}[descr/.style={fill=white}]
  \matrix (m) [matrix of math nodes, row sep=2em, column sep=3.5em]
{
    & X & & X \\
    A\times T & & A\times T & \\
    & P \times T & & P \times T \\
    & T &  & T  \\
};
  \path[-stealth,->,>=angle 90,font=\scriptsize]
    (m-1-2) edge node[descr,name=G] {$\Shift{q}\times\Shift{q}\times\Shift{q}$} (m-1-4)
            edge [red] node[above] {$\inp{A}$} (m-2-1)
            edge [right hook->,red,dotted] node[auto] {$\sched$} (m-3-2)
    (m-1-4) edge [right hook->,red] node[right] {$\sched$} (m-3-4)
            edge [red] node[above] {$\inp{A}$} (m-2-3)
    (m-2-1) edge node[descr,name=H] {$(\Shift{q}\times\Shift{q})\times(\intp{t})$} (m-2-3)
            edge [right hook->,red] node[below] {$\loc{A}$} (m-3-2)
    (m-2-3) edge [right hook->,red] node[below] {$\loc{A}$} (m-3-4)
    (m-3-2) edge node[descr,name=P] {$\torus{q}{2}\times(\intp{t})$} (m-3-4)
            edge node[right] {$\pi_T$} (m-4-2)
    (m-4-2) edge node[descr] {$\intp{t}$} (m-4-4)
    (m-3-4) edge node[right] {$\pi_T$} (m-4-4);

\draw[->,color=blue,dotted]
(G) --  node[left,near start,font=\scriptsize] {$\homsched$} (P);
\draw[->,color=blue] 
(G) -- node[above,font=\scriptsize] {$\hominp$} (H);
\draw[->,color=blue]
(H) -- node[below,font=\scriptsize] {$\homloc$} (P);

\node[draw,blue,rectangle] at ([yshift=1.5em]current bounding box.south west){$\homsched=\homloc\circ\hominp$};
\end{tikzpicture}
  \caption{Matrix Multiplication of size $n\times n\times n$ on a 2D
    torus network of size $q\times q$ with one copy of input set $A$.
    $\loc{A}$ and $f$ are embeddings.  Similar equations can be
    written for input set $B$ and output set $C$.  }
  \label{fig:mm-2d}
\end{figure}
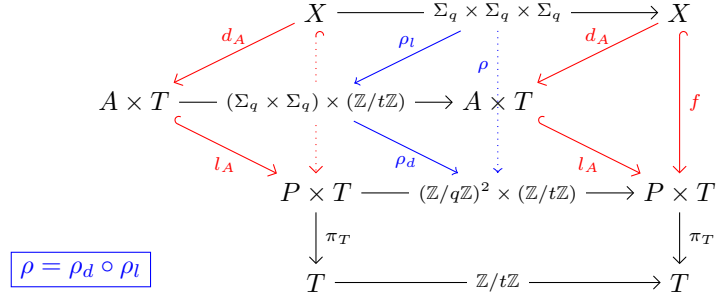

The homomorphism $\homsched:\Shift{q}\times\Shift{q}\times\Shift{q}
\rightarrow \torus{q}{2} \times\intp{t}$ is determined by the images
of generators $\shift$ in each $\Shift{q}$.  Let $g_x=(1,0),g_y=(0,1)$
be the generators of $\torus{q}{2}$, and let $\delta_t$ which
represents one increment in time be the generator of $\intp{t}$. Since
fixing the image of generators of the domain $\Shift{q}^3$ fixes the
rest of the homomorphism, suppose that
\begin{eqnarray*}
\homsched 
&:\sigma_1:=(\shift, e_{\Shift{q}}, e_{\torus{q}{2}}) &\mapsto ((x_1 g_x + y_1 g_y), t_1\delta_t)\\  
&:\sigma_2:=(e_{\Shift{q}}, \shift, e_{\torus{q}{2}}) &\mapsto ((x_2 g_x + y_2 g_y), t_2\delta_t)\\  
&:\sigma_3:=(e_{\Shift{q}}, e_{\Shift{q}}, \shift) &\mapsto ((x_3 g_x + y_3 g_y), t_3\delta_t).
\end{eqnarray*}

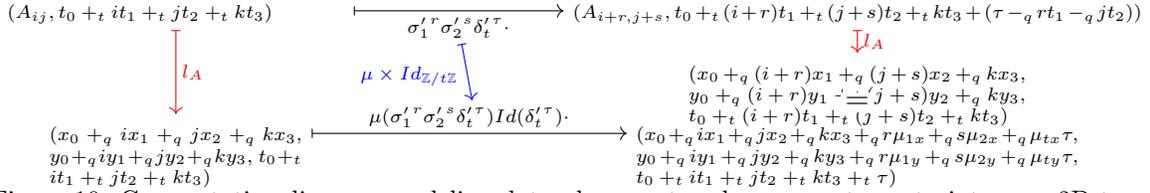
\begin{figure*}[th!]
  \begin{center}
\begin{tikzcd}[row sep=0.9em,column sep=10em,font=\scriptsize]
    \parbox[t]{16em}{$(A_{ij}, t_0\add{t}it_1\add{t}jt_2\add{t}kt_3)$}
      \arrow[red,mapsto]{dd}{\loc{A}}
      \arrow[mapsto]{r}[below,name=gact]{{\sigma'_1}^{r}{\sigma'_2}^{s}{\delta'_t}^{\tau}\cdot}& 
    \parbox[t]{27em}{$(A_{i+r,j+s}, t_0\add{t}(i+r)t_1\add{t}(j+s)t_2
                       \add{t}kt_3+(\tau\subt{q}rt_1\subt{q}jt_2))$}
      \arrow[red,mapsto]{d}{\loc{A}} \\
      &
    \parbox[t]{16em}{$(x_0 \add{q} (i+r)x_1 \add{q} (j+s)x_2 \add{q} kx_3$,
      $y_0 \add{q} (i+r)y_1 \add{q} (j+s)y_2 \add{q} ky_3,$
      $\ \ \ t_0\add{t}(i+r)t_1\add{t} (j+s)t_2\add{t}kt_3)$}
      \arrow[-,dotted]{d}[description,font=\large]{=} \\
    \parbox[t]{12em}{$(x_0 \add{q} ix_1 \add{q} jx_2 \add{q} kx_3$,
      $y_0 \add{q} iy_1 \add{q} jy_2 \add{q} ky_3,$
      $t_0\add{t}it_1\add{t}jt_2\add{t}kt_3)$}
      \arrow[mapsto]{r}[name=hact]
            {\mu({\sigma'_1}^{r}{\sigma'_2}^{s}{\delta'_t}^{\tau})Id({\delta'_t}^{\tau})\cdot} &
    \parbox[t]{21em}{$(x_0 \add{q} ix_1 \add{q} jx_2 \add{q} kx_3
      \add{q} r\mu_{1x}\add{q} s\mu_{2x} \add{q}\mu_{tx}\tau,$
      $y_0 \add{q} iy_1 \add{q} jy_2 \add{q} ky_3
      \add{q} r\mu_{1y}\add{q} s\mu_{2y} \add{q}\mu_{ty}\tau,$
      $t_0\add{t}it_1\add{t}jt_2\add{t}kt_3\add{t}\tau)$}
\arrow[mapsto,blue,to path={([xshift=10pt]gact) -- node[left] {$\mu\times Id_{\intp{t}}$} (hact)}]{}
\end{tikzcd}
\label{fig:comm-sol-tori}
\caption{Commutative diagram modeling data placement and movement
  constraints on a 2D torus.}
  \end{center}
\end{figure*}

The homomorphism $\homsched$ with respect to which the schedule
$\sched$ is equivariant fixes $\sched$ for every $x\in X$ if we pick
its value at one point. If we choose $ \sched: X_{000} \mapsto
(x_0,y_0,t_0)$,
\footnote{using the notation $(x_0,y_0)$ for $P_{x_0,y_0}$} 
then $\sched$ maps $X_{ijk}=\sigma_1^i\sigma_2^j\sigma_3^k\cdot
X_{000} \mapsto
\homsched(\sigma_1)^i\homsched(\sigma_2)^j\homsched(\sigma_3)^k
\cdot(x_0,y_0,t_0)=$
\[
(x_0 \add{q} ix_1\add{q} jx_2\add{q}kx_3,
    y_0 \add{q} iy_1\add{q} jy_2\add{q}ky_3,
 t_0 \add{t} it_1\add{t} jt_2\add{t}kt_3).
\]
Since the input variable $A_{ij}$ is required for the instruction
$X_{ijk}$, the choice of $\sched(X_{000})$ forces for all $k\in[q]$
the map
\[
\loc{A}:(A_{ij},t_0 \add{t} it_1\add{t} jt_2\add{t}kt_3) \mapsto
\homsched(\sigma_1)^i\homsched(\sigma_2)^j\homsched(\sigma_3)^k
\cdot(x_0,y_0,t_0).
\]

Let $\homloc$ be the homomorphisms corresponding to the placement of
the input set $A$.  Recall that $\homloc$ is of the form $ ( \mu:
\Shift{q}^2 \times \intp{t} \rightarrow \torus{q}{2} ) \times
Id_{\intp{t}}. $ A generator set for the domain of $\homloc$ consists
of the two one-step cyclic shifts of $\Shift{q}^2$ and the one time
step increment $\delta'_t$.  Suppose that the image of $\homloc$ on
the generators is
\begin{eqnarray*}
\mu: & \sigma'_1=(\shift, e_{\Shift{q}} , e_{\intp{t}})
&\mapsto \mu_{1x}g_x+\mu_{1y}g_y\\
\mu: & \sigma'_2=(e_{\Shift{q}}, \shift , e_{\intp{t}})
 &\mapsto \mu_{2x}g_x+\mu_{2y}g_y\\
\mu: & \delta'_t=(e_{\Shift{q}}, e_{\Shift{q}}, \delta_t)
 &\mapsto \mu_{tx}g_x +\mu_{ty}g_y.
\end{eqnarray*}

This gives us the diagram in Fig.~\ref{fig:comm-sol-tori} which
commutes only when $\tau=rt_1 \add{t} st_2$, and
$r(x_1\subt{q}\mu_{1x} \add{q}\mu_{tx}t_1) \add{q}
s(x_2\subt{q}\mu_{2x}\add{q}\mu_{tx}t_2)$ and $r(x_1\subt{q}\mu_{1x}
\add{q}\mu_{tx}t_1) \add{q} s(x_2\subt{q}\mu_{2x}\add{q}\mu_{tx}t_2)$
vanish for all $r,s\in\ints$. This represents a set of equations,
which hold when
\[
 \left( \begin{array}{cc}
\mu_{1x}\subt{q}x_1 & \mu_{1y}\subt{q}y_1 \\
\mu_{2x}\subt{q}x_2 & \mu_{2y}\subt{q}y_2 
\end{array} \right)
\equiv_q
 \left( \begin{array}{c}
t_1 \\
t_2
\end{array} \right)
\cdot
 \left( \begin{array}{cc}
\mu_{tx} &  \mu_{ty}
\end{array} \right).\] 

For the equivariant maps $\sched$ and $\loc{A}$ to be embeddings, the
images of the homomorphisms $\homsched$ and $\homloc$ must be at least
$q^3$ and $q^2t$ in size respectively.  Therefore, the group generated
by $((x_i g_x + y_i g_y),t_i\delta_t)$ for $i=1,2,3$ must be
isomorphic to $\torus{q}{2}\times\intp{q}$, and the group generated by
$(\mu_{ix}g_x+\mu_{iy}g_y)$ for $i=1,2$ must be isomorphic to
$\torus{q}{2}$.

The minimum number of steps for $\sched$ to be an embedding is $t=q$.
The communication cost associated with the homomorphism $\mu$ is the
number of hops associated with the network element
$(\mu_{tx},\mu_{ty})$ multiplied by the number of variables $p^2$ in
the set $A$ multiplied by the number of time steps $p$. Ideally,
$|\mu_{tx}|+|\mu_{ty}|$ is $1$ or $0$. Note that this can not be the
case for all three of $A,B,C$; the movement cost factor determined by
$\mu$ can be $0$ for at most one of them. The Cannon's algorithm
(Fig.~\ref{fig:cannon}) corresponds to the values
\[
\left| \begin{array}{ccc}
x_1 & y_1 & t_1 \\
x_2 & y_2 & t_2 \\
x_3 & y_3 & t_3 
\end{array}
=
\begin{array}{ccc}
1 & 0 & -1 \\
0 & 0 & 1 \\
0 & 1 & -1
\end{array}\right|
\left|
\begin{array}{cc}
\mu_{1x} & \mu_{1y} \\
\mu_{2x} & \mu_{2y} \\
\mu_{tx} & \mu_{ty} 
\end{array}
=
\begin{array}{cc}
1 & 1 \\
0 & 1 \\
0 & 1
\end{array}
\right|,
\]
where $\mu$ corresponds to the input set $A$.  The values for $B$ and
$C$ can be similarly calculated.  Other Cannon-like algorithms can be
obtained by setting the variables in the above
$(x_i,y_i,t_i)_{i=1,2,3}$ and $(\mu_{ix}, \mu_{iy})_{i=1,2}$ matrices
to make them unimodular (determinant $\pm 1$). The row and
column-permutation flexibility in the Cannon's algorithm arises from
the choice of assigning variables indices $i,j,k$ to the rows and
columns in the data.

One might observe that the above calculations have the flavor of
linear maps between vector spaces \cite{KL83} and unimodular
transformations such as in \cite{WL91-TPDS}.  This is expected since
the group modeling the network here is also a field.

When $l,m,n\geq q$, a blocked version of the Cannon's algorithm can be
derived for the $q\times q$-size 2D-torus machine.  Suppose that
$l=q\cdot q_l,m=q\cdot q_m$ and $n=q\cdot q_n$. The subgroup
$(\perm{q_l}\wr\perm{q})$ of $\perm{l}$ has the subgroup
$(\perm{q_l}\wr\Shift{q})$ where $\Shift{q}$ acts on the $q$
permutation blocks of size $\perm{q_l}$ by shifting them cyclically.
We have the transitive subgroups $\perm{q_l}\wr\Shift{q}
\leq\perm{q_l}\wr\perm{q}\leq \perm{l}$. Similar transitive subgroups
can be listed for $\perm{m}$ and $\perm{n}$. A homomorphism
$\rho':\Shift{q}^3\rightarrow \torus{q}{2}\times\intp{t}$ can be
extended to a homomorphism 
\[
\homsched:(\perm{q_l}\wr\Shift{q})\times(\perm{q_m}\wr\Shift{q})
\times(\perm{q_n}\wr\Shift{q}) \rightarrow\torus{q}{2}\times\intp{t}
\]
by augmenting $\rho'$ with the projection of the subgroups
$\perm{q_l},\perm{q_m}$ and $\perm{q_n}$ to the identity
$e_{\torus{q}{2}\times\intp{t}}$.  A schedule equivariant with such a
$\homsched$ forces the homomorphism $\homloc$ to map at least $q_l
q_m$ elements of $A$ to each node at any point in time.  Therefore, if
we allow at least $q_l q_m + q_m q_n + q_n q_l$ memory on each node
for the blocks of the input and the output data, the machinery
developed for the case $l=m=n=q$ extends to yield a blocked version of
Cannon-like algorithms.

A blocked version of Cannon-like algorithm does not minimize
communication for all values of parameters.  For $n\times n\times
n$-sized matrix multiplication, the communication induced by
blocked-Cannon's algorithms on $\sqrt{p}\times\sqrt{p}$ processors is
$3\sqrt{p}\times p\times n^2/p = 3n^2\sqrt{p}$, which amounts to
$3n^2/\sqrt{p}$ per node.  When the memory available per node $M$ is
more than what is required to store one copy of $A,B$ and $C$ across
$p$ nodes, i.e.  $p M \geq 3n^2$, the communication cost per node
$3n^2/p$ is greater than the lower bound $3n^3/\sqrt{M}$
\cite{ITT04,HBL}.

When there is sufficient space for $c$ copies of each variables in
$A,B$ and $C$: $p M\geq c\times 3n^2$, taking advantage of extra
memory to replicate $A,B$ and $C$ variable sets $c$-fold can reduce
communication to the minimum required.  Let $q=\sqrt{p/c}$.  The
``2.5D-algorithm'' \cite{SD11} achieves optimality on a 3D-toroidal
network of dimensions $q \times q \times c$ with the network group
$\intp{q} \times \intp{q} \times \intp{c}$.  The schedule (i)
partitions the torus into $c$ 2D-tori layers of size $q \times q$,
(ii) assigns one copy of $A,B$ and $C$ to each of the $c$ layers,
(iii) maps the variables to the nodes in blocks of size $n\sqrt{c}/p
\times n\sqrt{c}/p$, (iv) performs $t$-skewed steps of the Cannon's
schedule in each sub-grid.  See section~\ref{sec:2-5D} for a detailed
derivation of this schedule.

\input{recursion}

%% file: recursion.tex
\subsection{Fat-trees: recursive schedules}

We will now consider machines on which recursive versions of matrix
multiplication are useful.  We will start with $4\times 4\times
4$-size matrix multiplication on the fat-tree with $4$ processors.  As
described in Section~\ref{sec:machine-model}, the network is modeled
by the group $N:=\iterwr{2}{2}$.  Suppose also that each processing
node has three units of memory, one for each copy of a variable from
the sets $A,B,C$.  At least two time steps are needed for completion.
Let $T=\{t_i\}_{i\in[2]}$ and $\Delta=\intp{2}$.  The schedule is
determined by $\sched(X_{000})$ and the homomorphism from
$G=\perm{2}\times\perm{2}\times\perm{2}$ to the group $N\times\tinc =
\iterwr{2}{2}\times\intp{2} \cong ((\perm{2}\times\perm{2})
\rtimes\perm{2}) \times \intp{2}$.  The group $G$ is generated by
$\sigma_i,\sigma_j$ and $\sigma_k$ which flip the indices $i,j$ and
$k$ respectively. The group $N$ is generated by $\sigma_{00}$ and
$\sigma_{01}$ which swap the processors in the left and right halves
of the machine (lower permutations), and $\sigma_{10}$ which swaps the
left and the right pair of processors (the higher permutation).  Let
$\sigma_t$ be the generator of $\tinc$.  Since $\sched$ is an
embedding (at most one instruction can be computed at a processor at
each time step), the image $\homsched(G)$ of $G$ should be isomorphic
to $G$. For this, the image of each of $\sigma_i,\sigma_j,\sigma_k$
with respect to the lower permutations should be the same.  Further,
at least one of the three images should affect each of the lower and
upper permutations in $N$ and the generator of $\Delta$.  Such
homomorphisms are easily enumerated. One of them is described in
Fig.~\ref{fig:rec2}.
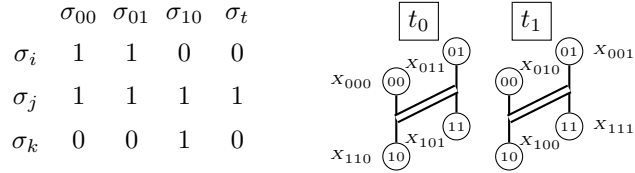
\begin{figure}[!htb]
  \begin{center}
\hspace{-15pt}
\begin{tikzpicture}
  \matrix (template) [matrix of nodes,nodes={inner sep=0pt,text width=20pt,align=center,minimum height=.5cm}]{
             & $\sigma_{00}$ & $\sigma_{01}$ & $\sigma_{10}$ & $\sigma_{t}$ \\
  $\sigma_i$ & 1 & 1 & 0 & 0 \\
  $\sigma_j$ & 1 & 1 & 1 & 1 \\
  $\sigma_k$ & 0 & 0 & 1 & 0 \\
 };
\end{tikzpicture}\hspace{15pt}
\begin{tikzpicture}
\tikzstyle{vertex}=[circle,draw,minimum size=6pt,inner sep=1pt,font=\tiny]
\tikzstyle{leftsmall}=[left,font=\tiny]
\tikzstyle{rightsmall}=[right,font=\tiny]
\tikzstyle{edge} = [draw,thick,-,black]
\tikzstyle{invisible} = [draw=none,minimum size=0pt]
\tikzstyle{doubleedge} = [draw,thick,double distance=2pt,black]

\node[vertex] (x00) at (0,1) {${00}$};
\node[leftsmall] at (x00.west) {$X_{000}$};
\node[vertex] (x01) at (0.8,1.4) {${01}$};
\node[leftsmall] at (x01.south) {$X_{011}$};
\node[vertex] (x10) at (0,0) {${10}$};
\node[leftsmall] at (x10.west) {$X_{110}$};
\node[vertex] (x11) at (0.8,0.4) {${11}$};
\node[leftsmall] at (x11.south) {$X_{101}$};
\draw[edge] (x00) -- (x10);
\draw[edge] (x01) -- (x11);
\tikzstyle{invisible} = [draw=none,minimum size=0pt]
\tikzstyle{doubleedge} = [draw,thick,double distance=2pt,black,xslant=0.4]
\node[invisible] (x0) at (0,0.5) {};
\node[invisible] (x1) at (0.8,0.9) {};
\draw[doubleedge] (x0.center) -- (x1.center);

\node[vertex] (xt00) at (1.5,1) {${00}$};
\node[rightsmall] at (xt00.north) {$X_{010}$};
\node[vertex] (xt01) at (2.3,1.4) {${01}$};
\node[rightsmall] at (xt01.east) {$X_{001}$};
\node[vertex] (xt10) at (1.5,0) {${10}$};
\node[rightsmall] at (xt10.north) {$X_{100}$};
\node[vertex] (xt11) at (2.3,0.4) {${11}$};
\node[rightsmall] at (xt11.east) {$X_{111}$};
\draw[edge] (xt00) -- (xt10);
\draw[edge] (xt01) -- (xt11);
\tikzstyle{invisible} = [draw=none,minimum size=0pt]
\tikzstyle{doubleedge} = [draw,thick,double distance=2pt,black,xslant=0.4]
\node[invisible] (xt0) at (1.5,0.5) {};
\node[invisible] (xt1) at (2.3,0.9) {};
\draw[doubleedge] (xt0.center) -- (xt1.center);

\node[draw,rectangle] () at (0.3,1.8) {$t_0$};
\node[draw,rectangle] () at (1.8,1.8) {$t_1$};
\end{tikzpicture}
  \caption{The homomorphism $\homsched$ on the left. One indicates
    that the generator corresponding to the row flips the permutation
    in the column and zero indicates that the column element is
    unaffected. The corresponding schedule $f$ on the right. The index
    in the circle represents the processor number. The two slices
    correspond to the two time steps $t_0$ and $t_1$.  }
\label{fig:rec2}
\end{center}
\end{figure}

With the choice of $\sched(X_{000})=(P_{00},t_0)$, the homomorphism
$\homsched$ fixes the schedule $\sched$ for all other $X_{ijk}$ as in
Fig.~\ref{fig:rec2} (right).  The data placement and movement forced
by this schedule is illustrated below in Fig.~\ref{fig:rec-data}.

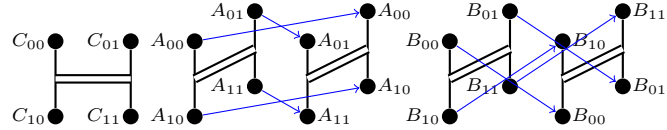
\begin{figure}[!htb]
  \begin{center}
\hspace{-15pt}
\begin{tikzpicture}
\tikzstyle{vertex}=[circle,fill=black,minimum size=6pt,inner sep=0pt,]
\tikzstyle{leftsmall}=[left,font=\scriptsize]
\tikzstyle{edge} = [draw,thick,-,black]
\tikzstyle{invisible} = [draw=none,minimum size=0pt]
\tikzstyle{doubleedge} = [draw,thick,double distance=2pt,black]

\node[vertex] (c00) at (0,1) {};
\node[leftsmall] at (c00) {$C_{00}$};
\node[vertex] (c01) at (1,1) {};
\node[leftsmall] at (c01) {$C_{01}$};
\node[vertex] (c10) at (0,0) {};
\node[leftsmall] at (c10) {$C_{10}$};
\node[vertex] (c11) at (1,0) {};
\node[leftsmall] at (c11) {$C_{11}$};
\draw[edge] (c00) -- (c10);
\draw[edge] (c01) -- (c11);
\node[invisible] (c0) at (0,0.5) {};
\node[invisible] (c1) at (1,0.5) {};
\draw[doubleedge] (c0.center) -- (c1.center);
\end{tikzpicture}
\hspace{-6pt}
\begin{tikzpicture}
\tikzstyle{vertex}=[circle,fill=black,minimum size=6pt,inner sep=0pt]
\tikzstyle{leftsmall}=[left,font=\scriptsize]
\tikzstyle{rightsmall}=[right,font=\scriptsize]
\tikzstyle{edge} = [draw,thick,-,black]
\tikzstyle{invisible} = [draw=none,minimum size=0pt]
\tikzstyle{doubleedge} = [draw,thick,double distance=2pt,black]

\node[vertex] (a00) at (0,1) {};
\node[leftsmall] at (a00) {$A_{00}$};
\node[vertex] (a01) at (0.8,1.4) {};
\node[leftsmall] at (a01) {$A_{01}$};
\node[vertex] (a10) at (0,0) {};
\node[leftsmall] at (a10) {$A_{10}$};
\node[vertex] (a11) at (0.8,0.4) {};
\node[leftsmall] at (a11) {$A_{11}$};
\draw[edge] (a00) -- (a10);
\draw[edge] (a01) -- (a11);
\tikzstyle{invisible} = [draw=none,minimum size=0pt]
\tikzstyle{doubleedge} = [draw,thick,double distance=2pt,black,xslant=0.4]
\node[invisible] (a0) at (0,0.5) {};
\node[invisible] (a1) at (0.8,0.9) {};
\draw[doubleedge] (a0.center) -- (a1.center);

\node[vertex] (at00) at (1.5,1) {};
\node[rightsmall] at (at00) {$A_{01}$};
\node[vertex] (at01) at (2.3,1.4) {};
\node[rightsmall] at (at01) {$A_{00}$};
\node[vertex] (at10) at (1.5,0) {};
\node[rightsmall] at (at10) {$A_{11}$};
\node[vertex] (at11) at (2.3,0.4) {};
\node[rightsmall] at (at11) {$A_{10}$};
\draw[edge] (at00) -- (at10);
\draw[edge] (at01) -- (at11);
\tikzstyle{invisible} = [draw=none,minimum size=0pt]
\tikzstyle{doubleedge} = [draw,thick,double distance=2pt,black,xslant=0.4]
\node[invisible] (at0) at (1.5,0.5) {};
\node[invisible] (at1) at (2.3,0.9) {};
\draw[doubleedge] (at0.center) -- (at1.center);

\draw[blue,->] (a00) -- (at01);
\draw[blue,->] (a01) -- (at00);
\draw[blue,->] (a10) -- (at11);
\draw[blue,->] (a11) -- (at10);
\end{tikzpicture}
\hspace{-16pt}
\begin{tikzpicture}
\tikzstyle{vertex}=[circle,fill=black,minimum size=6pt,inner sep=0pt]
\tikzstyle{leftsmall}=[left,font=\scriptsize]
\tikzstyle{rightsmall}=[right,font=\scriptsize]
\tikzstyle{edge} = [draw,thick,-,black]
\tikzstyle{invisible} = [draw=none,minimum size=0pt]
\tikzstyle{doubleedge} = [draw,thick,double distance=2pt,black]

\node[vertex] (b00) at (0,1) {};
\node[leftsmall] at (b00) {$B_{00}$};
\node[vertex] (b01) at (0.8,1.4) {};
\node[leftsmall] at (b01) {$B_{01}$};
\node[vertex] (b10) at (0,0) {};
\node[leftsmall] at (b10) {$B_{10}$};
\node[vertex] (b11) at (0.8,0.4) {};
\node[leftsmall] at (b11) {$B_{11}$};
\draw[edge] (b00) -- (b10);
\draw[edge] (b01) -- (b11);
\tikzstyle{invisible} = [draw=none,minimum size=0pt]
\tikzstyle{doubleedge} = [draw,thick,double distance=2pt,black,xslant=0.4]
\node[invisible] (b0) at (0,0.5) {};
\node[invisible] (b1) at (0.8,0.9) {};
\draw[doubleedge] (b0.center) -- (b1.center);

\node[vertex] (bt00) at (1.5,1) {};
\node[rightsmall] at (bt00) {$B_{10}$};
\node[vertex] (bt01) at (2.3,1.4) {};
\node[rightsmall] at (bt01) {$B_{11}$};
\node[vertex] (bt10) at (1.5,0) {};
\node[rightsmall] at (bt10) {$B_{00}$};
\node[vertex] (bt11) at (2.3,0.4) {};
\node[rightsmall] at (bt11) {$B_{01}$};
\draw[edge] (bt00) -- (bt10);
\draw[edge] (bt01) -- (bt11);
\tikzstyle{invisible} = [draw=none,minimum size=0pt]
\tikzstyle{doubleedge} = [draw,thick,double distance=2pt,black,xslant=0.4]
\node[invisible] (bt0) at (1.5,0.5) {};
\node[invisible] (bt1) at (2.3,0.9) {};
\draw[doubleedge] (bt0.center) -- (bt1.center);

\draw[blue,->] (b00) -- (bt10);
\draw[blue,->] (b10) -- (bt00);
\draw[blue,->] (b01) -- (bt11);
\draw[blue,->] (b11) -- (bt01);
\end{tikzpicture}  
  \caption{Data movement implied by Fig.~\ref{fig:rec2}}
  \label{fig:rec-data}
\end{center}
\end{figure}

The data movement homomorphism $\mu_C:\iterwr{2}{2}\times\tinc
\rightarrow \iterwr{2}{2}$ corresponding to $C$ is the trivial map to
the identity of $\iterwr{2}{2}$. No data is moved between the time
steps.  The elements in $A$ are swapped across the higher level
permutation between time steps:
\[
\mu_A:(((e_{00},e_{01}),\sigma_{10}),\sigma_t) \mapsto
((e_{00},e_{01}),\sigma_{10}).
\]
Therefore four elements of $A$ cross the higher-level connection and
two elements traverse each of the four lower-level connections. The
elements in $B$ are swapped across the lower level permutation between
time steps
\[
\mu_B:(((\sigma_{00},\sigma_{01}),e_{10}),\sigma_t) \mapsto
((\sigma_{00},\sigma_{01}),e_{10}).
\]
This corresponds to moving two elements on each of the lower-level
connections. In total, $4$ units of data are moved across the
higher-level link and $8$ units are moved across the lower-level
links. This is the minimum amount of communication required for
classical matrix multiplication on this topology with the memory
limitations described.

We will inductively construct schedules for $n\times n\times n$-matrix
multiplication on a $n^2$-size fat- for any $n=2^d$.  As before,
assume that each node has only three units of memory, one for a
variable each from the sets $A,B,C$.

The network group is $N_d=\iterwr{2}{2d}$.  We pick the transitive
subgroup $\iterwr{2}{d}$ of $\perm{n}$ so that the subgroup of
$\sym{X}$ we choose is $G_d=\iterwr{2}{d}\times \iterwr{2}{d} \times
\iterwr{2}{d}$.  $G_d$ is also constructed from
$\perm{2}\cong\intp{2}$ Like the network group.  Therefore for
homomorphism from $G_d$, it is superfluous to construct $\tinc$ from
any group other than $\perm{2}$ (doing so would use only few of the
time increments in the group).  Since we need at least $n=2^d$
different time steps, let $\tinc=\tinc_d=\perm{2}^d$ and
$T=\prod_{i=1}^{d} \timepr{i}$, $\timepr{l}=\{t_{l,0},t_{l,1}\}$.

Suppose we have a schedule and the corresponding homomorphism
$\homsched_{d-1}: (\iterwr{2}{d-1})^3 \rightarrow \iterwr{2}{2d-2}
\times \perm{2}^{d-1}$ for $n=2^{d-1}$.  To construct $\homsched_d$,
we put together $\homsched_{d-1}$, and $\homsched$ for the base case
$d=1$ discussed earlier.  For this, first note that $N_d$ is
isomorphic to $N_{d-1}\wr \iterwr{2}{2}$ which is the two-level
permutation $\iterwr{2}{2}$ acting on four instances of $N_{d-1}$.
Further, $\tinc_d \cong \tinc_{d-1}\times\intp{2}$.

A generator set of $G_{d-1}$ consists of the union of three generators
sets $\chi_i,\chi_j,\chi_k$ of the three $\iterwr{2}{d-1}$
corresponding to indices $i,j,k$.  Since $G_d \cong
(((\iterwr{2}{d-1})^2)\rtimes \perm{2})^3$, a generator set of $G_d$
is two copies of the generator set of $G_{d-1}$,
$\{\chi_{i1},\chi_{i2},\chi_{j1},\chi_{j2},\chi_{k1},\chi_{k2}\}$
augmented with the generators $\sigma_{i,d},\sigma_{j,d},\sigma_{k,d}$
of the three top-level permutations.  Since $N_d\times \tinc_d \cong
(N_{d-1}^4\rtimes\iterwr{2}{2}) \times(\tinc_{d-1}\times\intp{2})$, a
generator set for it is four copies
$\{\chi_{n1},\chi_{n2},\chi_{n3},\chi_{n4}\}$ of the generator set
$\chi_n$ of $N_{d-1}$, one copy of the generator set $\chi_{t,d-1}$ of
$\tinc_{d-1}$ augmented with the generators
$\{\sigma_{n,00},\sigma_{n,01},\sigma_{n,11}\}$ of the top two-level
permutation $\iterwr{2}{2}$, and the generator $\sigma_{t,d}$ of
$\intp{2}$.

To construct $\homsched_{d}$, we map $\chi_{i1}$ to $\chi_{n1}$ and
$\chi_{t,d-1}$ in the same way that $\chi_i$ is mapped to $\chi_{n}$
and $\chi_{t,d-1}$ in $\homsched_{d-1}$. We map other copies of the
generator sets similarly. We map the top-level generators
$\{\sigma_{i,d},\sigma_{j,d},\sigma_{k,d}\}$ of $G_d$ to the top-level
generators of $N_d\times \tinc_d$,
$\{\sigma_{n,00},\sigma_{n,01},\sigma_{n,11}, \sigma_{t,d}\}$ as in
Fig.~\ref{fig:rec2}.

The corresponding schedule is the recursive schedule that splits $A,B$
and $C$ into four quadrants, splits the machine in to four quadrants
at the two highest levels of the hierarchy, splits the time steps into
two parts based on $\timepr{d}$ and does the eight smaller matrix
multiplications on the eight-fold partition of $P\times T$ in the same
way as in Fig.~\ref{fig:rec2}. The smaller matrix multiplications are
scheduled in the subsets of $P\times T$ similarly. The schedule never
moves $C$, moves $n^2$ amount of data corresponding to $A$ across the
highest $2d$-level connection and moves $2n^2$ amount of data
corresponding to $A$ and $B$ across the $(2d-1)$-level links. This is
the minimum amount of communication required for this machine.  To
compute the time, $\tinc$ can be flattened appropriately, with the
stretch of each time step at a particular level reflecting the time
required for communication at the level in the hierarchy.

\subsection{Caches and Parallel Memory Hierarchies}

An $(h+1)$-level inclusive parallel memory hierarchy \cite{ACF93} (see
Fig.~\ref{fig:pmh}) is a tree of caches with $(h+1)$ levels indexed by
$[h+1]$.  Each node at level-$i$ represents a cache of size $M_i$, and
has $f_i$ caches beneath it at level-$(i-1)$. The root of the tree
represents the main memory (level $h$) and the processors are the
$p=f_h f_{h-1} \dots f_1$ leaves at level $0$.  Suppose that (i) for
each $i>0$, $M_i = 3\cdot2^{d_i}$ and $f_i = 2^{c_i}$ for some
positive even integers $c_i<d_i$, (ii) each processors has three
registers so that $M_0=3$, and (iii) the cache lines are one word
long.

We will model this machine with $M_h/3$ nodes, each node with $3$
words of memory.  All the nodes collectively represent the
all-inclusive level-$h$ cache.  Of these, the first $f_h M_{h-1}/3$
nodes represent the $f_h$ level-$(h-1)$ caches. In each of these $f_h$
blocks of size $M_{h-1}/3$, the first $f_{h-1}M_{h-2}/3$ are
identified with level $f_{h-1}$ of level-$(h-2)$ caches, and so on for
each level. The network group is inductively defined:
$N_1=\perm{M_1/3}$ and $N_i=N_{i-1}\wr \perm{M_i/M_{i-1}}$, where the
action of the wreath product can be seen as allowing the permutation
of the contents of the $M_i/M_{i-1}$ partitions of level-$i$ cache,
with $f_i$ of the partitions representing the level-$(i-1)$.  Let $T =
\prod_{i=1}^{h} T_i$, where $|T_i| =
t_i=2^{3((d_i-d_{i-1})/2-c_i/3)}$, and $\tinc = \tinc_h$ where
$\tinc_l=\prod_{i=1}^{l}\intp{t_i}$.

As before we will construct the schedule and the homomorphism
$\homsched$ by induction on the levels of the (cache) hierarchy.
Consider $2^{d_i/2}\times 2^{d_i/2} \times 2^{d_i/2} $-size
multiplication with instruction set $X_i$, the largest that fits in
a level-$i$ cache.  For $i=1$,
$G_1=(\Shift{2^{c_1/3}}\wr\iterwr{2}{(d_1/2-c_1/3)})^3$ is a
transitive subgroup of $\sym{X_1}$. $N_1$ has the subgroup
$\Shift{2^{c_1}}\wr\iterwr{2}{d_1-c_1}$.  We choose the homomorphism
$\homsched_1:G_1\rightarrow N_1\times\tinc_1$ that (i) maps
the three $\Shift{2^{c_1/3}}$ in $G_1$ on to $\Shift{2^{c_1}}$ in
$N_1$ by flattening them, and (ii) lifts this map along $\tinc_1$
with the rest of the $3(d_1/2-c_1/3)$ of the top-level $\perm{2}$
permutations by mapping each one to a unique $2^{t'} \in \tinc_1$
according to its order in the iterated wreath product, with the lower
level permutations assigned to smaller $t'$. The schedule for this
corresponds to a Z-order traversal in time of $X_1$ with
$2^{c_1/3}\times 2^{c_1/3} \times 2^{c_1/3}$ size blocks of $X_1$
executed on the $2^{c_i}$ processors in each time step.  Note that
$G_i= (\Shift{2^{c_1/3}}\wr\iterwr{2}{(d_1/2-c_2/3)}) \wr
(\Shift{2^{c_2/3}}\wr\iterwr{2}{((d_2-d_1)/2-c_2/3)}) \wr \dots \wr
(\Shift{2^{c_i/3}}\wr\iterwr{2}{((d_i-d_{i-1})/2-c_i/3)}) )^3 $ is a
transitive subgroup of $\sym{X_i}$. We can extend
$\homsched_{i-1}:G_{i-1}\rightarrow N_{i-1}\times\tinc_{i-1}$ to
$\homsched_i:G_i\rightarrow N_i\times\tinc_i$ using the construction
for $i=1$.

A schedule that is equivariant with such $\homsched_i$ is a
\textbf{space-bounded schedule} \cite{CSBR13,BFGS11} for the parallel
recursive matrix multiplication algorithm \cite{BGS10}.  A
space-bounded scheduler is based on the principle of executing tasks
in a fork-join program (such as in recursive matrix multiplication) on
the processors assigned to the lowest level cache the task fits
in. This minimizes communication at all levels in the hierarchy.  When
all $f_i=1$ and we have one processor, this corresponds to the
execution of the recursive algorithm on a sequential machine with an
$h$-level hierarchy of ideal caches \cite{FLPR99}.

%% file: conclusion.tex
\section{Limitations and Future Work}
We have demonstrated a technique to develop schedules for an algorithm
on different machines using matrix multiplication as an example.
These schedules were time- and communication-optimal on the machines
considered. The technique involves solving commutative diagrams using
knowledge of the subgroup structure of the symmetries of
the algorithm, and optimizing over the time and communication costs
associated with possible homomorphisms to the groups representing the
network and time increments.

However effective it might be for the example at hand, it is to be
noted that we have applied this technique ``by hand'' without
addressing the computational complexity of this technique.  It is
worth considering whether this procedure can be efficiently performed
using computational algebra packages \cite{GAP}.

Classical matrix multiplication is about the easiest example to
demonstrate the efficacy of this technique. This does not imply this
technique works for every algorithm.  In fact, several algorithms do
not allow symmetry preserving schedule for certain topologies.
 Such algorithms might need a partitioning of $X$ and a separate
 symmetry-preserving map to schedule each partition.  However, matrix
 multiplication is a good start since it is the building block of
 numerical linear algebra.  To model these algorithms, an immediate
 goal would be to extend the model to handle dependencies.

Relevant directions for further work include integration of the model
with (a) models for irregularity \cite{BFGS11}, (b) non-constant
replication factors such as in the SUMMA algorithm, (c) higher level
abstractions for communication and space tradeoffs, (d) communication
lower bounds \cite{HBL}; and extending the model for (e) asynchronous
time steps, and (f) overlapping computation and
communication \cite{LU-overlap}.


%% file: appendix.tex
\section{Some Preliminary definitions}
\label{sec:prelims}
\begin{definition}[Group]
A group $G$ is a set with a special ``identity'' element $e_G$ along
with an associative binary operator '+' on the set with respect to
which the set is closed: $g,h\in G \implies g+h\in G$. Further, for
each element $g\in G$: (i) $e_G+g=g+e_G=g$ and (ii) there exists an
inverse denoted by $-g$ such that $g+(-g)=(-g)+g=e$.  Often the binary
operator is denoted by $\cdot$ or not explicitly written when clear
from context. For $k\in \ints, g\in G$, we denote $g+g+.._k+g =:kg$
(or $g^k$ if we use $\cdot$).  A subset $S\subseteq G$ is said to
\defn{generate} $G$ if all elements of $G$ can be expressed as the
combination of finitely many elements of $S$. A subset $K$ of $G$ that
is itself a group w.r.t. the '+' operator and identity $e_G$ is called
a \defn{subgroup} of $G$, and denoted by $K\leq G$. The \textbf{direct
  product} $G\times H$ of two groups $G$ and $H$ is the group defined
by the set $G\times H$ with the operation
$(g,h)\cdot(g',h')=(gg',hh')$ for all $g,g'\in G,\ h,h'\in H$ and the
identity $(e_G,e_H)$.
\end{definition}

\begin{definition}[Homomorphism]
 A homomorphism from the group $G$ to the group $H$ is a function
 $\rho:G\rightarrow H$ s.t. (i) $\rho(e_G)=e_H$, and (ii)
 $\rho(g_1g_2) =\rho(g_1)\rho(g_2)$ for all $g_1,g_2\in G$.  A
 homomorphism is completely fixed by the image of a generator set of
 $G$.  An isomorphism is a bijective homomorphism.  If it exists, we
 say $G$ is isomorphic to $H$: $G\cong H$.
\end{definition}

\begin{definition}[(left) group action]
A left group action of a group $G$ on a set $X$ is a set of functions
that map $X \rightarrow X$ parameterized by the elements $g$ of the
group $G$, and denoted by $g\cdot$, such that (i) $e_G\cdot$ is the
identity map $Id_X$, and (ii) $(gh)\cdot x = g\cdot(h\cdot x)$ for all
$g,h\in G$ and $x\in X$.  We drop the adjective left in the rest of
this paper. When a group action of the group $G$ on the set $X$ is
clear from the context, we represent it diagrammatically as follows.
\begin{center}
\begin{tikzcd}[row sep=0em,column sep=9em]
  X \arrow{r}[description]{G} & X
\end{tikzcd}
\end{center}
\end{definition}

\section{Additional Diagrams}
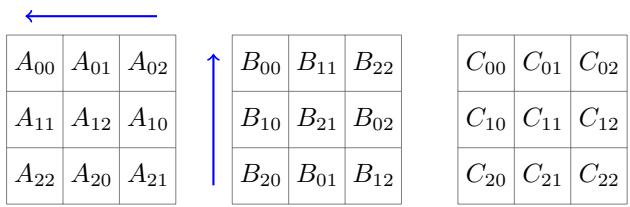
\begin{figure}[!htb]
  \begin{center}
  \begin{tikzpicture}
\draw[step=.75,gray,very thin] (0,0) grid (2.25,2.25);
\draw[step=.75,gray,very thin] (2.999,0) grid (5.25,2.25);
\draw[step=.75,gray,very thin] (5.999,0) grid (8.25,2.25);

\matrix[matrix of nodes,
inner sep=0pt,
anchor=south west,
nodes={inner sep=0pt,text width=.75cm,align=center,minimum height=.75cm}
] at (0,0) {
 $A_{00}$ & $A_{01}$ & $A_{02}$ \\
 $A_{11}$ & $A_{12}$ & $A_{10}$ \\
 $A_{22}$ & $A_{20}$ & $A_{21}$ \\
};

\matrix[matrix of nodes,
inner sep=0pt,
anchor=south west,
nodes={inner sep=0pt,text width=.75cm,align=center,minimum height=.75cm}
] at (3,0) {
 $B_{00}$ & $B_{11}$ & $B_{22}$ \\
 $B_{10}$ & $B_{21}$ & $B_{02}$ \\
 $B_{20}$ & $B_{01}$ & $B_{12}$ \\
};

\matrix[matrix of nodes,
inner sep=0pt,
anchor=south west,
nodes={inner sep=0pt,text width=.75cm,align=center,minimum height=.75cm}
] at (6,0) {
 $C_{00}$ & $C_{01}$ & $C_{02}$ \\
 $C_{10}$ & $C_{11}$ & $C_{12}$ \\
 $C_{20}$ & $C_{21}$ & $C_{22}$ \\
};

\draw[blue,->,thick] (2,2.5) -- (0.25,2.5);
\draw[blue,->,thick] (2.75,0.25) -- (2.75,2);
  \end{tikzpicture}
  \caption{Cannon's Algorithm for a $3\times 3$-torus depicting the
    data placement at $t_0$ and movement between time steps (in blue).
    Each element of $A$ moves one step left in each time step, and
    each element of $B$ moves one step up each time step. $C$ remains
    at the same location.}
  \label{fig:cannon}
    
  \end{center}
\end{figure}

\begin{figure}[!htb]
\centering
  \includegraphics[width=0.49\textwidth]{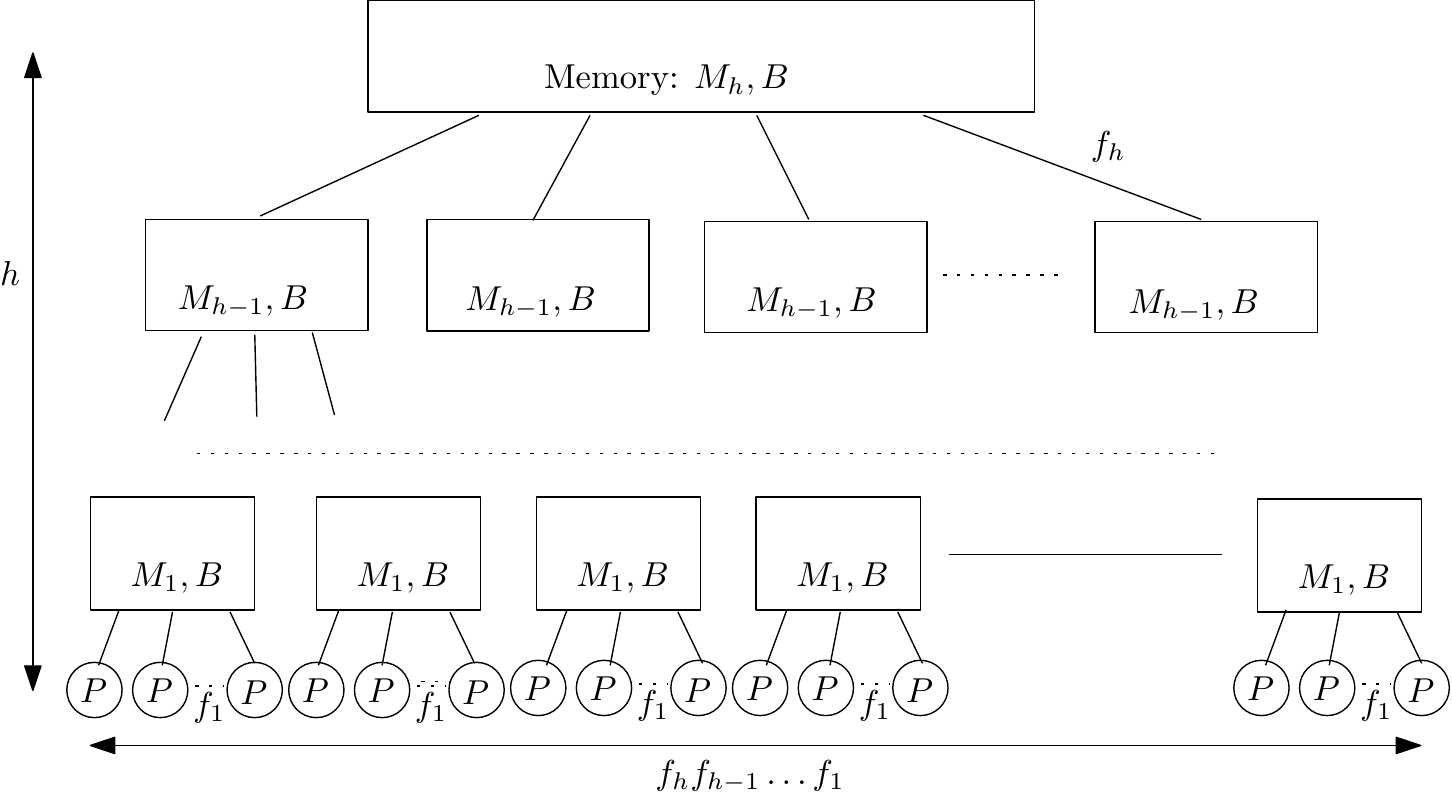}
  \caption{A Parallel Memory Hierarchy \cite{ACF93}. $M_i$ represent
    the sizes of the cache, $B$ represent cache line size. We set
    $B=1$ (this can be relaxed). $f_i$ represents the fan-out at each
    level. Number of processors, which are at the leaves, is $f_h
    f_{h-1}\dots f_1$.}
  \label{fig:pmh}
\end{figure}

\section{Proofs}
\label{sec:proofs}

\begin{proof} [of lemma~\ref{lem:coset-maps}]
If $\alpha$ exists, $\alpha(glL)=\alpha(gL)$ for all $l\in L,g\in G$.
Suppose $\alpha(L)=aK$. Then, $glaK = gaK$ which implies
$a^{-1}laK=K$, and therefore, $a^{-1}la \in K$ or $L^a\subseteq
K$. The converse follows from the construction of the $G$-equivariant
map $\hat{a}:gL\mapsto gaK$. For $l\in L$, $\hat{a}(glL)
=glaK=ga(a^{-1}la)K=gaK = \hat{a}(gL)$ since $a^{-1}la\in K$.  The
uniqueness of this construction follows from the equivariance.
\end{proof}

\begin{proof} [of Lemma~\ref{lem:imp-hom}]
Suppose the image of $\sigma\in G$ is $\rho(\sigma)=x\neq
e_{\intp{q}}$.  Suppose its cycle decomposition contains
$(k_1,k_2,\dots,k_c)$-size cycles with $k_1+k_2+\dots+k_c=q$ and
$k_i<q$. Then $\sigma'=\sigma^{k_1 k_2\dots k_c}=e_G$.  Therefore,
$\rho(\sigma')=(k_1k_2\dots k_c)x = e_{\intp{q}}$. Since $x$ is drawn
from a cyclic group, this can happen only if $q | k_1 k_2\dots k_c$.
This is a contradiction since $k_i$ are smaller than $q$, and $q$ is a
prime.
\end{proof}

\begin{proof} [of Lemma~\ref{lem:pp-hom}]
First note that since $q$ is prime, $\rho$ is non-trivial, and the
primitive $\sigma\notin\ker\rho$, $\sigma^k\in\ker\rho$ if and only if
$k$ is a multiple of $q$.  Also, $\rho(\sigma^k) = e_{\intp{q}}$ if
and only if $k$ is a multiple of $q$. Further
$\sigma^{-1}=\sigma^{q-1}$ and $\rho(\sigma^{q-1})=
\rho(\sigma^{-1})\neq \rho(\sigma)$.

Suppose $\sigma'\in G$ is imprimitive; then so is
$\sigma\cdot\sigma'$.  By lemma~\ref{lem:imp-hom}, both $\sigma'$ and
$\sigma\cdot\sigma'$ are in $\ker\rho$ so that $\rho(\sigma)=
\rho(\sigma)\cdot\rho(\sigma')=
\rho(\sigma\cdot\sigma')=e_{\intp{q}}$, a contradiction. So
imprimitive permutations are not in $G$.

Suppose $\sigma''\in G$ is primitive, but there does not exist an
integer $k$ such that $\sigma^k=\sigma''$. Then $\sigma\cdot\sigma'$
is imprimitive and has a non-trivial cycle decomposition, say
$(y_1,y_2,\dots,y_c)$. Let $y=y_1 y_2\dots y_c$. Then,
$\rho(\sigma)^y\rho(\sigma'')^y=\rho(\sigma\cdot\sigma'')^y
=e_{\intp{q}}$. Therefore, $\rho(\sigma^{-1})^y=\rho(\sigma'')^y$.
Since $q$ does not divide $y$, we conclude that $\rho(\sigma'')
=\rho(\sigma^{-1})$. Since $\sigma^{-1}=\sigma^{q-1}$ is primitive and
$\sigma''\neq (\sigma^{q-1})^k$ for any integer $k$ , we could have
similarly concluded that $\rho(\sigma'')=\rho((\sigma^{-1})^{-1})
=\rho(\sigma)$. Therefore, we have $\rho(\sigma)=\rho(\sigma')$, a
contradiction.

So the only elements in $G$ are of the form $\sigma^k$ for some
integer $k$.
\end{proof}

\begin{proof} [of Lemma~\ref{lem:tp-hom}]
The group $G$ is cyclic by lemma~\ref{lem:pp-hom}.  The proof follows
directly from this fact.
\end{proof}

\section{Other schedules}
\subsection{``2.5D''-algorithm}
\label{sec:2-5D}

Let $q=\sqrt{p/c}$.  The ``2.5D''-algorithm \cite{SD11} achieves
communication optimality on a 3D-toroidal network of dimensions $q
\times q \times c$ with the network group $\intp{q} \times \intp{q}
\times \intp{c}$ with generators $g_x, g_y, g_z$ and identity element
$e_{3D}$. Suppose that $p$ is a multiple of $c^{3/2}$ and let
$t=q/c=p^{1/2}/c^{3/2}$. Then, the group $\perm{n}$ has the transitive
subgroup $\perm{n/ct}\wr (\perm{t}\wr \perm{c})$.  Further,
\[
\intp{t} \times\intp{c}
\cong\Shift{t}\times\Shift{c}
\leq \Shift{t}\wr \Shift{c}
\leq \perm{t}\wr \perm{c}.
\]
are all transitive subgroups of $\perm{tc}$. We will consider
homomorphisms from the subgroup $(\Shift{t}\times\Shift{c})$ which are
determined by the images of the generator $\shift^t$ of $\Shift{t}$
and the generator $\shift^c$ of $\Shift{c}$. Let $\rho':(\Shift{t}\times
\Shift{c})^3 \rightarrow (\torus{q}{2}\times\intp{c})\times \intp{t}$ be the
homomorphism
\begin{eqnarray*}
\rho'
&:((\shift^t,e_{\Shift{c}}),(e_{\Shift{t}},e_{\Shift{c}}),(e_{\Shift{t}},e_{\Shift{c}}))
&\mapsto (g_x,-\delta_t)\\
&:((e_{\Shift{t}},e_{\Shift{c}}),(\shift^t,e_{\Shift{c}}),(e_{\Shift{t}},e_{\Shift{c}}))
&\mapsto (e_{3D},\delta_t)\\
&:((e_{\Shift{t}},e_{\Shift{c}}),(e_{\Shift{t}},e_{\Shift{c}}),(\shift^t,e_{\Shift{c}}))
&\mapsto (g_y,-\delta_t)\\
&:((e_{\Shift{t}},\shift^c),(e_{\Shift{t}},e_{\Shift{c}}),(e_{\Shift{t}},e_{\Shift{c}}))
&\mapsto (e_{3D}, e_{\intp{t}})\\
&:((e_{\Shift{t}},e_{\Shift{c}}),(e_{\Shift{t}},\shift^c),(e_{\Shift{t}},e_{\Shift{c}}))
&\mapsto (g_z, e_{\intp{t}})\\
&:((e_{\Shift{t}},e_{\Shift{c}}),(e_{\Shift{t}},e_{\Shift{c}}),(e_{\Shift{t}},\shift^c))
&\mapsto (e_{3D}, e_{\intp{t}}).
\end{eqnarray*}  

The homomorphism $\homsched$ obtained by augmenting $\rho'$ with 
 the projection of $\perm{n/ct}$ to the identity
$e_{(\torus{q}{2}\times\intp{c})\times\intp{t}}$ corresponds to the
``2.5D''-schedule \cite{SD11}. The ``2.5D''-schedule (i) partitions
the torus into $c$ 2D-tori layers of size $q \times q$, (ii) assigns
one copy of $A,B$ and $C$ to each of the $c$ layers, (iii) maps the
variables to the nodes in blocks of size $n\sqrt{c}/p \times
n\sqrt{c}/p$, (iv) performs $t$-skewed steps of the Cannon's schedule
in each sub-grid.  The schedule must be supported by a suitable
replication at the beginning and a reduction of $C$ at the end.

\subsection{Hexagonal VLSI Array}

\begin{figure}[!htb]
\centering
  \includegraphics[width=0.49\textwidth]{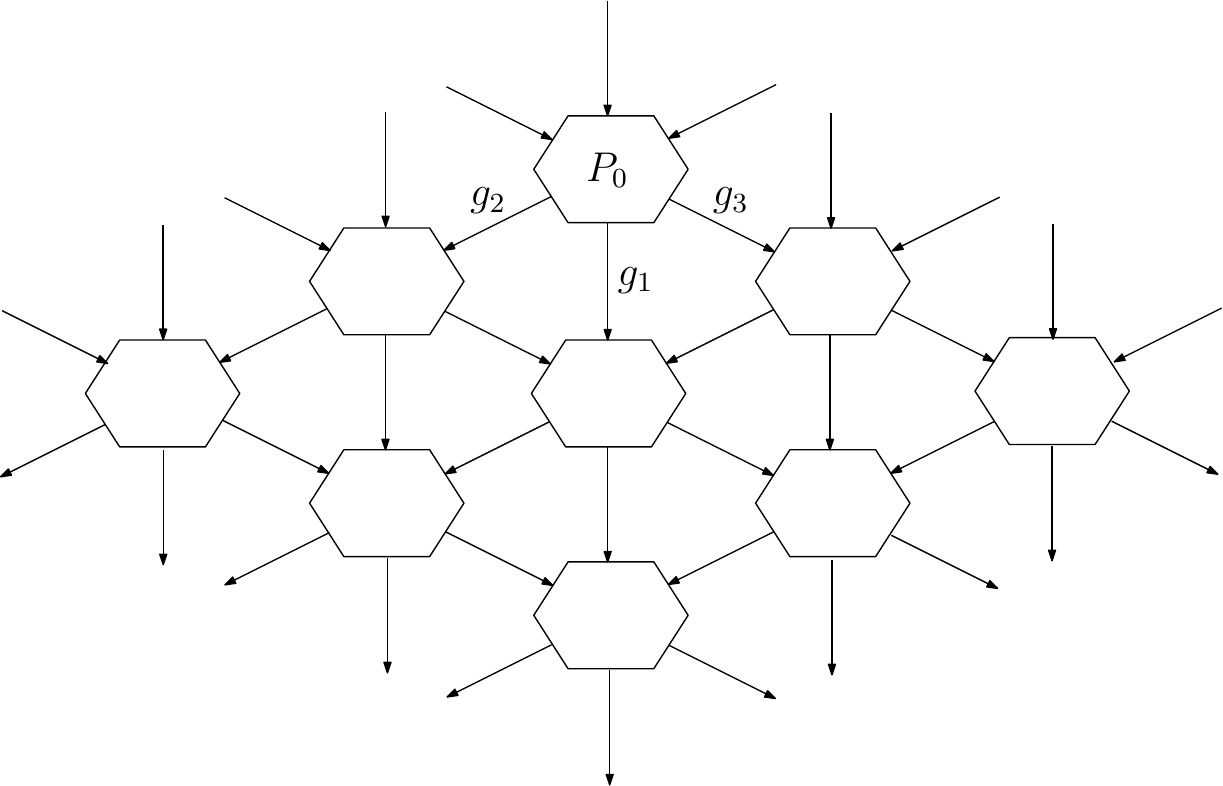}
  \caption{A Hexagonal VLSI array. Each node can multiply and
    accumulate. Nine nodes are shown.  The group action of generators
    of the network group is labeled on the node $P_0$. Their action on
    other nodes is similar. }
  \label{fig:hex}
\end{figure}

The hexagonal VLSI array of multiply and accumulate nodes
from~\cite{KungVLSIalgos} can be modeled as the action of the group
$N=\gen{g_1, g_2, g_3\ |\ g_ig_j=g_jg_i,\ g_1=g_2g_3}$ (the infinite
abelian group generated by $g_1,g_2$ and $g_3$ such that $g_1=g_2g_3$)
on the infinite array of nodes pictured in Fig.~\ref{fig:hex}.  

To schedule a $q\times q\times q$-size multiplication, we consider
the subgroup $G=\Shift{q}\times\Shift{q}\times\Shift{q}$ subgroup of
$\sym{X}$. Time steps are modeled as the action of $\tinc=\intp{3q}$
on $T=\{t_i\}_{i\in[3q]}$ (there are no embeddings of $X$ in $P\times
T$ for fewer time steps). Suppose that $\delta_t\in\tinc$ increments
time steps by one; it generates $\tinc$.

The homomorphism $\homsched:G\rightarrow N\times\tinc$ imposed by the
following images of the generator set $\{g_1,g_2,g_3\}$ of $G$
corresponds to the schedule in Figure 3 of ~\cite{KungVLSIalgos},
reproduced here in Fig.~\ref{fig:vlsihex}.

\begin{eqnarray*}
\homsched
:(\shift, e_{\Shift{q}},e_{\Shift{q}})
& \mapsto &(g_2, \delta_t)\\
:(e_{\Shift{q}},\shift, e_{\Shift{q}})
& \mapsto &(-g_1,\delta_t)\\
:(e_{\Shift{q}}, e_{\Shift{q}},\shift)
& \mapsto &(g_3,\delta_t)
\end{eqnarray*}

We simply have to anchor the schedule somewhere in the infinite array
by choosing some value for $\sched(X_{000})$ which fixes the rest of
the schedule. The corresponding homomorphism $\homloc$ can be easily
calculated. The homomorphism $\mu$ in $\homloc$ does not change with
time (as in Cannon's algorithm) and models the time-invariant
``direction, speed and timing'' of data movement referred to
in~\cite{KungVLSIalgos}.

\begin{figure}
\centering
  \includegraphics[width=0.49\textwidth]{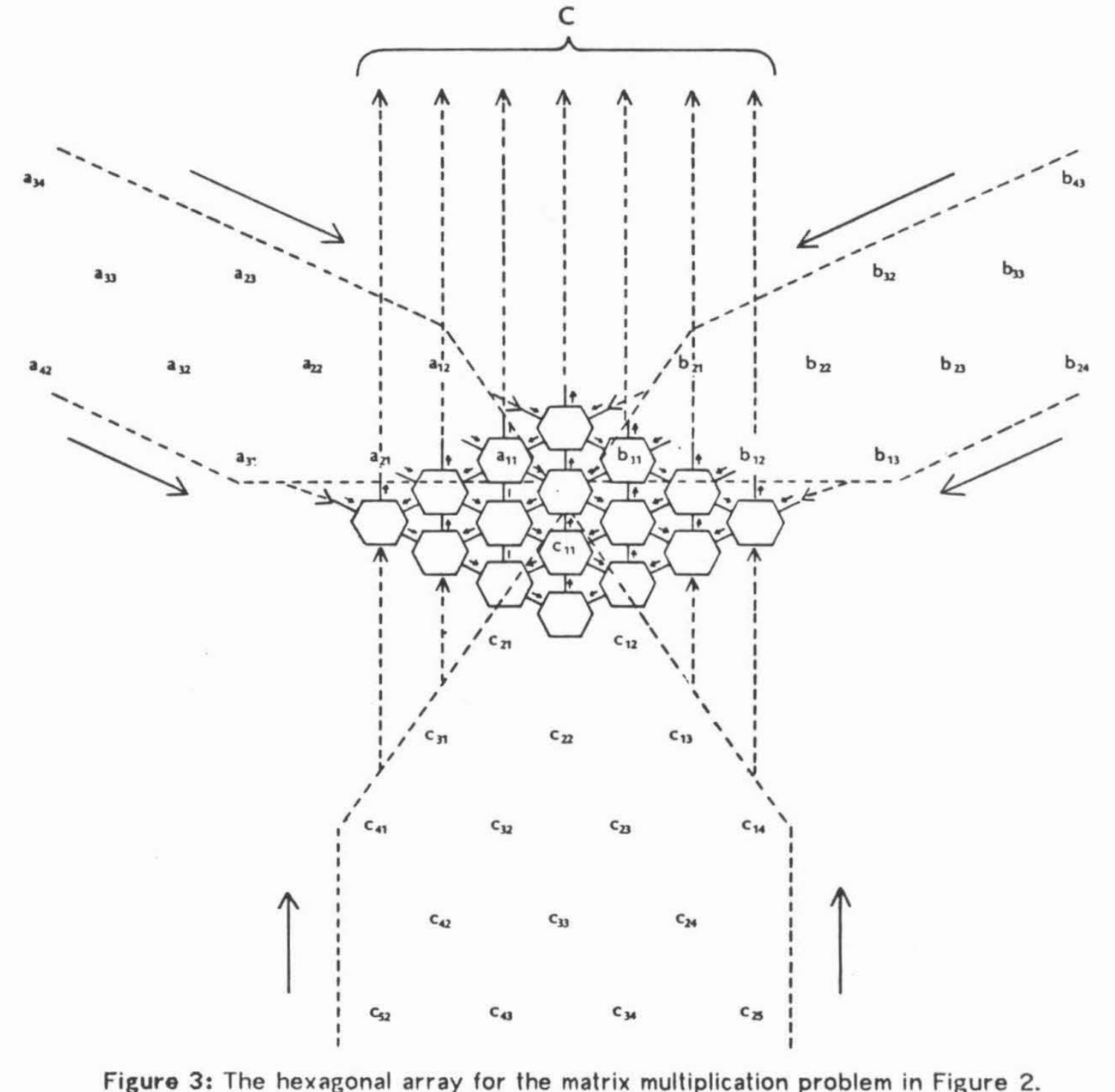}
  \caption{Systolic computation for matrix multiplication on VLSI
    array. Reproduced from ~\cite[Fig.~3]{KungVLSIalgos}.}
  \label{fig:vlsihex}
\end{figure}

\section{Acknowledgments}
I thank Katherine Yelick for motivating the central question of this
paper, and James Demmel and members of the BeBOP group at UC Berkeley,
including Evangelos Georganas, Penporn Koanantakool and Nicholas
Knight, for patiently listening to several iterations of these
ideas. I thank Joseph Landsberg whose course inspired this line of
thought.  I thank Niranjini Rajagopal for comments on the draft.  This
work is funded by the U.S. Department of Energy, Office of Science,
Office of Advanced Scientific Computing Research, Applied Mathematics
and Computer Science programs under contract No. DE-AC02-05CH11231,
through the Dynamic Exascale Global Address Space (DEGAS) programming
environments project.
